\documentclass[sigconf,screen,nonacm]{acmart}

\renewcommand\footnotetextcopyrightpermission[1]{}

\usepackage{amsmath,amsfonts,amsthm,mathtools,bm}
\usepackage{booktabs,tabularx}
\usepackage{enumitem}
\usepackage{microtype}
\usepackage{float}
\usepackage{placeins}
\usepackage{pgfplots}
\usepgfplotslibrary{groupplots}
\pgfplotsset{compat=1.18}
\pgfplotsset{
  caspaxis/.style={
    grid=both,
    major grid style={draw=gray!35, line width=0.22pt},
    minor grid style={draw=gray!20, line width=0.16pt},
    tick label style={font=\scriptsize},
    label style={font=\small},
    x label style={at={(axis description cs:0.5,1.04)}, anchor=south},
    title style={font=\small},
    line cap=round,
    line join=round,
  },
  caspgroup/.style={
    caspaxis,
    every axis plot/.append style={line width=1.15pt},
  },
}

\AtBeginDocument{%
  }

\newcommand{\E}{\mathbb{E}}
\newcommand{\Pbb}{\mathbb{P}}
\newcommand{\R}{\mathbb{R}}
\newcommand{\cX}{\mathcal{X}}
\newcommand{\cA}{\mathcal{A}}

\newcommand{\argmax}{\operatorname*{argmax}}

\newcommand{\DR}{\mathrm{DR}}
\newcommand{\IPS}{\mathrm{IPS}}

\newcommand{\CASP}{\textsc{CASP}}
\newtheorem{theorem}{Theorem}
\newtheorem{proposition}[theorem]{Proposition}
\newtheorem{corollary}[theorem]{Corollary}
\theoremstyle{definition}
\newtheorem{definition}[theorem]{Definition}
\newtheorem{assumption}[theorem]{Assumption}
\theoremstyle{remark}
\newtheorem{remark}[theorem]{Remark}

\title{\CASP: Support-Aware Offline Policy Selection for Two-Stage Recommender Systems}

\author{Nilson Chapagain}
\affiliation{%
  \institution{Texas A\&M University}
  \city{College Station}
  \state{TX}
  \country{USA}
}
\email{nchapagain@tamu.edu}

\begin{document}

\begin{abstract}
Two-stage recommender systems first choose a candidate generator and then rank items within the generated set. Because the generator decides which items are available to the ranker, changing the generator changes both the policy value and the data support used to estimate that value. This creates an offline selection problem that standard single-stage objectives do not capture: a policy may look good under a retrieval score or a raw off-policy value estimate, but still be unreliable if it depends on weakly supported generator--item pairs. We propose \CASP{} (Coupled Action-Set Pessimism), a support-aware offline selector for finite libraries of two-stage recommender policies. \CASP{} combines doubly robust value estimation with a support-burden penalty. We show that stagewise rules that ignore downstream continuation value can be arbitrarily suboptimal, and we derive population, finite-class, and reconstructed-propensity guarantees for conservative selection. In simulations and a reconstructed \texttt{MovieLens 1M} application, \CASP{} selects lower-burden policies when estimated value and support credibility are in tension.
\end{abstract}


\maketitle

\section{Introduction}

Consider a marketplace or media recommender that swaps one candidate generator for another. The incumbent generator surfaces mostly popular items and looks strong under recall-oriented diagnostics. The challenger exposes rarer items that the downstream ranker converts into better engagement, margin, or long-term satisfaction. If we judge stage~1 only by upstream retrieval metrics, or if we learn the two stages separately, we can prefer the wrong generator because stage~1 does not merely reweight downstream rewards: it changes the set of items the ranker is allowed to choose from.

That dependence is not an implementation nuisance. It is a structural property of modern recommender-system design. Large-scale recommenders commonly use multi-stage pipelines in which a retrieval or candidate-generation module first narrows a very large item universe, and a downstream ranking module then optimizes among the exposed candidates \citep{covington2016youtube,chen2019topk}. This architecture is computationally necessary, but it also means that the first stage determines which downstream actions can even be considered. Classical recommender objectives, including matrix-factorization and implicit-feedback ranking losses, are useful for prediction and representation learning \citep{hu2008implicit,rendle2009bpr}, but they do not by themselves answer the deployment question raised by logged two-stage systems: whether a new generator--ranker policy is both valuable and evaluable under the support created by the logging policy.

Standard off-policy evaluation and learning methods from contextual bandits provide strong foundations for value estimation under logged feedback \citep{li2011offline,dudik2014doublyrobust,swaminathan2015crm,swaminathan2015jmlr,wang2017optimal}. Related counterfactual learning-to-rank and recommender-system work emphasizes that recommendation data are biased by the exposure process and that offline objectives must correct for the policy that generated the data \citep{bottou2013counterfactual,joachims2017unbiasedltr,schnabel2016recommendations,gilotte2018offline,agarwal2018cltr,joachims2018deep}. However, much of this literature treats the action set as fixed once the context is observed, or focuses on correcting exposure bias for a single ranking or recommendation decision. In a two-stage recommender, the overlap condition is itself shaped upstream: a generator can deprive the downstream ranker of the very items that would have made a policy valuable.

Existing two-stage recommender work already shows that candidate generation and ranking interact \citep{ma2020twostage}, and recent candidate-generator evaluation work argues that retrieval should be judged by downstream utility rather than upstream retrieval metrics alone \citep{wang2025candidategen}. Work on slate recommendation and sequential recommendation further shows that structured recommendation outputs create difficult counterfactual evaluation problems because the value of an action may depend on other exposed actions or later user responses \citep{swaminathan2017slate,mcinerney2020sequentialslate,ie2019slateq}. Work on deficient support and large action spaces shows that weak overlap can make value estimates unstable even before adding two-stage architectural coupling \citep{sachdeva2020deficientsupport,saito2022embeddings,saito2023offcem,sachdeva2024policyconvolution}. What remains open is an offline selection framework that treats stage-induced feasible support as a first-class object in both learning and deployment-facing model selection.

We study offline policy selection for exactly this setting. Stage~1 chooses a candidate generator from a finite library; stage~2 chooses an item from the induced candidate set. The object of interest is not only the value of a policy, but also whether that value is credible under the support created by the logged two-stage system. This distinction is important for intelligent recommender systems because offline model selection is often used as a deployment gate: an engineer or analyst must decide whether a new generator--ranker policy is safe enough to test, not merely whether it has the largest offline score.

We propose \CASP{}, for \emph{Coupled Action-Set Pessimism}. The method is coupled because the generator determines downstream feasibility, and it is pessimistic because unsupported or weakly supported policy choices should be discounted before deployment. In operational terms, \CASP{} acts as a support-aware decision layer: it takes a logged two-stage recommender pipeline, a candidate policy library, reward and propensity estimates, and returns both a selected policy and a support report describing how much the policy depends on weak logged coverage.

The paper is deliberately scoped. We do not claim a general surrogate-consistency theory for all two-stage policy classes, nor do we claim that conservative selection should always maximize raw offline value. Instead, we study a finite policy-library setting that is common in recommender-system experimentation, where teams compare a manageable set of generator and ranker configurations before deployment. This setting is closely related to counterfactual model selection in recommender systems, but the key difficulty is different from ordinary exposure correction: the candidate generator changes the feasible downstream action set, so weak support is created by the architecture itself. Within this setting, the paper asks a practical question: when offline value and support credibility disagree, how should the system select a policy?

This article makes four contributions:
\begin{enumerate}[leftmargin=1.5em]
    \item We formulate support-aware offline policy selection for two-stage recommender systems, where the first-stage generator induces the feasible action set available to the downstream ranker.
    \item We show that continuation-blind stagewise optimization can fail structurally: if the first-stage proxy ignores downstream continuation value, it can select a generator that is arbitrarily suboptimal in end-to-end value.
    \item We introduce \CASP{}, a conservative selector that combines doubly robust value estimation with a support-burden penalty. The burden term is tied to the second moment of the coupled importance ratio, so it directly measures how much a candidate policy relies on weak logged support.
    \item We evaluate the method in simulations and in a reconstructed semi-synthetic \texttt{MovieLens 1M} pipeline. The empirical claim is intentionally modest: \CASP{} is not a universal raw-value winner, but it traces a lower-burden learned frontier and is most useful when high estimated value and support credibility are in tension.
\end{enumerate}

\section{Problem Formulation}

We observe i.i.d.\ interactions
\[
O_i = (X_i, A_{1i}, A_{2i}, Y_i), \qquad i=1,\dots,n,
\]
where $X_i \in \cX$ is the request context, user state, or session context; $A_{1i} \in \cA_1=\{1,\dots,K_1\}$ is a stage-1 candidate-generator choice; and $A_{2i} \in \cA_2=\{1,\dots,K_2\}$ is the final chosen item. The stage-1 action induces a nonempty feasible set
\[
S(x,a_1)\subseteq \cA_2,
\]
and the second-stage choice must satisfy $A_2 \in S(X,A_1)$. Throughout this article, stage~1 is a finite generator/template choice and stage~2 is a single-item choice from the induced feasible set.

A feasible target policy is a pair $\pi=(\pi_1,\pi_2)$, where
\[
\pi_1(a_1 \mid x)
\qquad \text{and} \qquad
\pi_2(a_2 \mid x,a_1)
\]
satisfy $\pi_2(a_2 \mid x,a_1)=0$ whenever $a_2\notin S(x,a_1)$. The logged behavior policy is denoted by $\mu=(\mu_1,\mu_2)$ with the same support restriction. Let
\[
q(x,a_1,a_2)=\E[Y\mid X=x,A_1=a_1,A_2=a_2].
\]
Under sequential ignorability, consistency, and support compatibility, the end-to-end value of a target policy is
\begin{equation}
\label{eq:value}
V(\pi)
=
\E\Bigg[
\sum_{a_1\in\cA_1}\pi_1(a_1\mid X)
\sum_{a_2\in S(X,a_1)}\pi_2(a_2\mid X,a_1)q(X,a_1,a_2)
\Bigg].
\end{equation}

The key structural point is that stage 1 matters through both exposure and restriction. A candidate generator that looks strong under an upstream proxy can still be poor if it suppresses items that a downstream ranker would have used effectively. Conversely, a generator that looks weak under an isolated retrieval metric may be valuable if it exposes a feasible set that materially improves downstream value.

\paragraph{Running example.}
The generator-swap example from the introduction illustrates the main point: the relevant object is not a generator in isolation, but a coupled two-stage policy whose value and credibility both depend on the feasible support induced at stage~1.

\begin{assumption}[Finite actions and bounded rewards]
\label{ass:finite}
The action spaces $\cA_1$ and $\cA_2$ are finite, every feasible set $S(x,a_1)$ is nonempty, and $0\le Y\le M$ almost surely for some finite $M>0$.
\end{assumption}

\begin{assumption}[Sequential ignorability and consistency]
\label{ass:ignorability}
For every feasible pair $(a_1,a_2)$, $Y(a_1,a_2)\perp (A_1,A_2)\mid X$, and if $(A_1,A_2)=(a_1,a_2)$ then $Y=Y(a_1,a_2)$.
\end{assumption}

\begin{assumption}[Support compatibility]
\label{ass:overlap}
For every target policy $\pi$ under consideration and almost every $x$, whenever $\pi_1(a_1\mid x)\pi_2(a_2\mid x,a_1)>0$, we have $\mu_1(a_1\mid x)\mu_2(a_2\mid x,a_1)>0$. When a quantitative bound is needed, assume the logged propensities are bounded below by $\nu \in (0,1)$ on the relevant support.
\end{assumption}

\paragraph{Deployment target.}
The target of the paper is a deployment-facing selection problem rather than only an estimation problem. A recommender team has a finite library $\Pi$ of candidate two-stage policies, each corresponding to a feasible generator--ranker configuration. The offline system must return three objects:
\[
\widehat \pi \in \Pi,
\qquad
\widehat V(\widehat \pi),
\qquad
\widehat B(\widehat \pi),
\]
where $\widehat \pi$ is the selected policy, $\widehat V(\widehat \pi)$ is its estimated value, and $\widehat B(\widehat \pi)$ is a support diagnostic. A policy with high estimated value but very large support burden should be interpreted as a risky deployment candidate because its apparent value depends on parts of the generator--item space that were weakly represented in the logged data. This deployment interpretation is the reason the paper treats support as part of the model-selection objective rather than only as an after-the-fact diagnostic.

\section{Why Stagewise Optimization Can Miss End-to-End Utility}

The correct first-stage target is not an isolated retrieval score. It is the continuation value induced by the downstream feasible set.

\begin{definition}[Continuation value]
For any feasible second-stage rule $\pi_2$, define
\[
m_{\pi_2}(x,a_1)
=
\sum_{a_2\in S(x,a_1)} \pi_2(a_2\mid x,a_1)q(x,a_1,a_2),
\]
and define the pointwise optimal continuation value
\[
m^\star(x,a_1)=\max_{a_2\in S(x,a_1)} q(x,a_1,a_2).
\]
\end{definition}

\begin{proposition}[Correct stage-1 target]
\label{prop:dynamic}
Under Assumption~\ref{ass:finite}, for any fixed stage-1 rule $\pi_1$,
\[
\sup_{\pi_2}V(\pi_1,\pi_2)
=
\E\Bigg[\sum_{a_1\in \cA_1}\pi_1(a_1\mid X)m^\star(X,a_1)\Bigg].
\]
Consequently, a value-optimal first-stage rule can be chosen by maximizing $m^\star(x,a_1)$ pointwise in $a_1$.
\end{proposition}

Proposition~\ref{prop:dynamic} formalizes the core intuition behind the paper: stage 1 should be evaluated through the downstream feasible set it creates, not through a local retrieval score alone.

We now define a generic stagewise proxy rule.

\begin{definition}[Continuation-blind separable stagewise family]
Fix a stage-1 proxy score $g:\cX\times\cA_1\to\R$ and a stage-2 score $h$. A continuation-blind separable stagewise family first selects
\[
\pi_1^{\mathrm{sep}}(\cdot \mid x)
\in
\argmax_{\rho\in\Delta(\cA_1)}\sum_{a_1}\rho(a_1)g(x,a_1),
\]
and then, for each realized $(x,a_1)$, selects
\[
\pi_2^{\mathrm{sep}}(\cdot \mid x,a_1)
\in
\argmax_{\rho\in\Delta(S(x,a_1))}\sum_{a_2\in S(x,a_1)}\rho(a_2)h(x,a_1,a_2).
\]
We call the family continuation-blind because the first-stage score is driven by $g$ and need not agree with the continuation value $m^\star(x,a_1)$.
\end{definition}

\begin{theorem}[Class-level failure of continuation-blind separable rules]
\label{thm:stagewise}
Let
\[
A_g^\star(x_0)=\argmax_{a_1\in\cA_1} g(x_0,a_1).
\]
Suppose there exists a context $x_0$ and a stage-1 action $b\in \cA_1\setminus A_g^\star(x_0)$. Then for any $M>0$ there exists a two-stage recommendation problem satisfying Assumptions~\ref{ass:finite}--\ref{ass:overlap}, with rewards in $[0,M]$, for which every rule in the separable stagewise family puts first-stage mass only on $A_g^\star(x_0)$ at $x_0$, while every end-to-end value-optimal policy selects $b$ at $x_0$. Consequently, the resulting value gap can be made equal to $M$.
\end{theorem}

Theorem~\ref{thm:stagewise} upgrades the earlier pointwise counterexample into a class-level statement about continuation-blind separable rules. It does not claim that every conceivable end-to-end surrogate must fail. It says something narrower and more useful for this paper: once stage~1 controls downstream feasibility, any family whose first-stage Bayes decision is driven by a proxy ordering rather than continuation value can be made arbitrarily wrong at the population level.

\section{\CASP: Coupled Action-Set Pessimism}

\CASP{} is a support-aware selector for a finite library of two-stage recommender policies. The method starts from a simple operational concern: a high offline value estimate is not enough if the candidate policy places substantial mass on generator--item pairs that the logged system rarely, or never, exposed. The support burden defined below measures this reliance on weak logged support.

\begin{definition}[Support burden]
For any target policy $\pi$ and feasible triple $(x,a_1,a_2)$, define the coupled importance ratio
\[
w_\pi(x,a_1,a_2)=
\frac{\pi_1(a_1\mid x)\pi_2(a_2\mid x,a_1)}
     {\mu_1(a_1\mid x)\mu_2(a_2\mid x,a_1)}.
\]
Define the conditional burden
\[
B_\mu(\pi;x)=
\sum_{a_1\in\cA_1}\sum_{a_2\in S(x,a_1)}
\frac{\pi_1(a_1\mid x)^2\pi_2(a_2\mid x,a_1)^2}
     {\mu_1(a_1\mid x)\mu_2(a_2\mid x,a_1)},
\]
and the global support burden $B_\mu(\pi)=\E[B_\mu(\pi;X)]$.
\end{definition}

The burden is not an arbitrary regularizer. It is the second moment of the coupled importance ratio induced by the two-stage logging process.

\begin{proposition}[Second-moment interpretation]
\label{prop:secondmoment}
Under Assumptions~\ref{ass:finite} and~\ref{ass:overlap},
\begin{align*}
\E[w_\pi(X,A_1,A_2)^2 \mid X] &= B_\mu(\pi;X), \\
\E[w_\pi(X,A_1,A_2)^2] &= B_\mu(\pi).
\end{align*}
Consequently, if
\[
\widehat V_{\IPS}(\pi)=\frac{1}{n}\sum_{i=1}^n w_\pi(X_i,A_{1i},A_{2i})Y_i,
\]
then, under $0\le Y\le M$,
\[
\operatorname{Var}\{\widehat V_{\IPS}(\pi)\}
\le
\frac{M^2}{n}B_\mu(\pi).
\]
\end{proposition}

Proposition~\ref{prop:secondmoment} is the main reason the penalty is tied to $B_\mu(\pi)$. In a single-stage bandit, this quantity reduces to the usual second moment of the importance weight. In the present two-stage recommender, the same object has a stronger architectural meaning: it measures how the first-stage generator choice changes the downstream support on which the second-stage ranker is evaluated.

To estimate policy value, we use a doubly robust score. Let
\[
m_r(x;\pi)=
\sum_{a_1\in\cA_1}\pi_1(a_1\mid x)
\sum_{a_2\in S(x,a_1)}\pi_2(a_2\mid x,a_1)r(x,a_1,a_2)
\]
for any regression function $r$. Given an estimated reward model $\hat q$, define the score
\begin{equation}
\label{eq:drscore}
\psi_\pi(O;\hat q)=
m_{\hat q}(X;\pi)
+
\frac{\pi_1(A_1\mid X)\pi_2(A_2\mid X,A_1)}
     {\mu_1(A_1\mid X)\mu_2(A_2\mid X,A_1)}
\bigl(Y-\hat q(X,A_1,A_2)\bigr),
\end{equation}
and the estimator
\[
\widehat V_{\DR}(\pi)=\frac{1}{n}\sum_{i=1}^n \psi_\pi(O_i;\hat q_i),
\]
where each $\hat q_i$ is evaluated out of sample. The empirical burden is
\[
\widehat B_\mu(\pi)=\frac{1}{n}\sum_{i=1}^n B_\mu(\pi;X_i).
\]
The \CASP{} score is
\begin{equation}
\label{eq:sampleobjective}
\widehat J_\lambda(\pi)=\widehat V_{\DR}(\pi)-\lambda \widehat B_\mu(\pi),
\end{equation}
and the selected policy is
\[
\widehat\pi_\lambda
\in
\argmax_{\pi\in\Pi}\widehat J_\lambda(\pi).
\]

\begin{figure}[tbp]
\centering
\fbox{%
\begin{minipage}{0.94\columnwidth}
\textbf{Algorithm 1: \CASP{} for support-aware two-stage policy selection}

\smallskip
\textbf{Input:} logged data $\{(X_i,A_{1i},A_{2i},Y_i)\}_{i=1}^n$, feasible-set map $S(x,a_1)$, finite policy library $\Pi$, reward fits $\hat q_i$, behavior or reconstructed propensities $\hat e_i$, and penalty grid $\Lambda$.

\smallskip
\begin{enumerate}[leftmargin=1.4em,itemsep=0.2em]
    \item For each policy $\pi\in\Pi$, compute the out-of-sample doubly robust value estimate $\widehat V_{\DR,\hat e}(\pi)$.
    \item For each policy $\pi\in\Pi$, compute the support burden $\widehat B_{\hat e}(\pi)$ over the induced feasible sets.
    \item For each $\lambda\in\Lambda$, compute
    \[
    \widehat J_\lambda(\pi)=
    \widehat V_{\DR,\hat e}(\pi)-\lambda \widehat B_{\hat e}(\pi).
    \]
    \item Select
    \[
    \widehat\pi_\lambda\in\argmax_{\pi\in\Pi}\widehat J_\lambda(\pi).
    \]
    \item Report the selected policy together with $\widehat V_{\DR,\hat e}$, $\widehat B_{\hat e}$, effective sample size, maximum observed weight, off-support mass, and generator-share changes.
\end{enumerate}
\end{minipage}}
\caption{\CASP{} as a deployment-facing decision layer. The method does not only rank policies by value; it returns a policy together with support diagnostics that indicate whether the offline value estimate is credible enough to act on.}
\label{fig:casp-algorithm}
\end{figure}

\paragraph{Choosing the penalty.}
The penalty $\lambda$ represents the decision maker's tolerance for support risk. In the experiments we report a fixed headline value for comparability, but the recommended deployment use is to inspect the full value--burden path over $\lambda\in\Lambda$ and choose the smallest positive penalty that removes policies with unacceptable support burden or off-support mass. Equivalently, if a platform has a burden threshold $\mathcal{B}_{\max}$, one may select
\[
\widehat\pi
\in
\argmax_{\pi\in\Pi:\,\widehat B_{\hat e}(\pi)\le \mathcal{B}_{\max}}
\widehat V_{\DR,\hat e}(\pi).
\]
The penalized and constrained views serve the same purpose: they make support credibility part of the selection rule rather than an informal post-hoc warning.

\section{Theoretical Guarantees}

This section gives scoped guarantees for the finite policy-library setting studied in the paper. The theory is not meant to establish a universal surrogate-consistency result for all two-stage recommender policies. Instead, it supports three claims needed for deployment-facing selection. First, Proposition~\ref{prop:dynamic} and Theorem~\ref{thm:stagewise} show why stagewise proxy optimization can fail when stage~1 controls downstream feasibility. Second, Proposition~\ref{prop:secondmoment} shows that the burden penalty is the second moment of the coupled importance ratio and therefore measures offline stability. Third, the finite-class results below show how conservative selection behaves when the candidate policy library is fixed in advance.

The first result is a population calibration bound for the penalized objective. It should be read as a sanity guarantee: if the value-optimal policy is already well supported, then using a small support penalty cannot create large value regret at the population level.

\begin{theorem}[Population guarantee for \CASP{}]
\label{thm:popregret}
Let
\begin{align*}
J_\lambda(\pi) &= V(\pi)-\lambda B_\mu(\pi), \\
\pi^\star &\in\argmax_{\pi\in\Pi}V(\pi), \\
\pi_\lambda &\in\argmax_{\pi\in\Pi}J_\lambda(\pi).
\end{align*}
Under Assumptions~\ref{ass:finite}--\ref{ass:overlap},
\[
V(\pi^\star)-V(\pi_\lambda)\le \lambda B_\mu(\pi^\star).
\]
\end{theorem}

The bound is intentionally simple. Its role is not to claim that \CASP{} dominates value-only selection, but to formalize the price of conservatism. If the best value policy has small burden, then the population cost of penalizing burden is small. If the best value policy has large burden, the bound allows a larger gap, which is exactly the regime where a deployment decision maker may prefer a lower-value but more credible policy.

The second guarantee concerns finite model selection. Let $\Pi$ be a finite candidate class and assume the logged propensities are known and bounded away from zero on the relevant support.

\begin{theorem}[Finite-class guarantee for \CASP{} selection]
\label{thm:finite}
Suppose $\Pi$ is finite, $\widehat V_{\DR}$ is computed with out-of-sample reward fits as defined above, and the quantitative version of Assumption~\ref{ass:overlap} holds with lower bound $\nu$. Then for any $\alpha\in(0,1)$, with probability at least $1-\alpha$,
\begin{align*}
\sup_{\pi\in\Pi}\bigl|\widehat V_{\DR}(\pi)-V(\pi)\bigr|
&\le M(1+\nu^{-2})\sqrt{\frac{\log(4|\Pi|/\alpha)}{2n}}, \\
\sup_{\pi\in\Pi}\bigl|\widehat B_\mu(\pi)-B_\mu(\pi)\bigr|
&\le \nu^{-2}\sqrt{\frac{\log(4|\Pi|/\alpha)}{2n}}.
\end{align*}
Consequently, the pessimistic selector
\[
\widehat \pi_\lambda \in \argmax_{\pi\in\Pi}\widehat J_\lambda(\pi)
\]
satisfies
\begin{align*}
V(\pi^\star)-V(\widehat \pi_\lambda)
&\le
\lambda B_\mu(\pi^\star) \\
&\quad + 2M(1+\nu^{-2})\sqrt{\frac{\log(4|\Pi|/\alpha)}{2n}} \\
&\quad + 2\lambda \nu^{-2}\sqrt{\frac{\log(4|\Pi|/\alpha)}{2n}}.
\end{align*}
with the same probability.
\end{theorem}

Theorem~\ref{thm:finite} is stated for the finite policy library used throughout the experiments. Although it does not address sharper infinite-class results with shared cross-fitted nuisances, it gives a direct conservative model-selection guarantee for the setting studied here.

Let
\[
e(x,a_1,a_2)=\mu_1(a_1\mid x)\mu_2(a_2\mid x,a_1),
\]
let $\hat e_i$ denote an out-of-sample or reconstructed estimate of the same joint propensity, and let $\widehat V_{\DR,\hat e}$ and $\widehat B_{\hat e}$ denote the empirical value and burden estimators formed by replacing $e$ with the observation-specific denominator $\hat e_i$. For fold- or observation-specific propensities, define the nuisance-conditional population burden average
\begin{align*}
B_{\hat e,n}(\pi)
\;=\;&
\frac{1}{n}\sum_{i=1}^n
\E\Bigg[
\sum_{a_1\in\cA_1}\sum_{a_2\in S(X,a_1)}
\frac{\pi_1(a_1\mid X)^2\pi_2(a_2\mid X,a_1)^2}
     {\hat e_i(X,a_1,a_2)}
\;\Bigg|\;\hat e_i
\Bigg],
\end{align*}
and define
\[
\Delta_V
=
\sup_{\pi\in\Pi}
\left|
\E[\widehat V_{\DR,\hat e}(\pi)\mid \{\hat q_i,\hat e_i\}_{i=1}^n]-V(\pi)
\right|,
\]
\[
\Delta_B
=
\sup_{\pi\in\Pi}
\bigl|B_{\hat e,n}(\pi)-B_\mu(\pi)\bigr|.
\]

\begin{corollary}[Finite-class \CASP{} guarantee with reconstructed propensities]
\label{cor:estimated}
Suppose $\Pi$ is finite, each evaluation observation is scored with out-of-sample $\hat q_i$ and $\hat e_i$, $0\le \hat q_i\le M$, and $\hat e_i\ge \hat\nu^2$ on the relevant support. Let
\[
\widehat\pi_{\lambda,\hat e}
\in
\argmax_{\pi\in\Pi}
\left\{
\widehat V_{\DR,\hat e}(\pi)-\lambda\widehat B_{\hat e}(\pi)
\right\}.
\]
Then, conditional on the nuisance fits, with probability at least $1-\alpha$,
\[
V(\pi^\star)-V(\widehat\pi_{\lambda,\hat e})
\le
\lambda B_\mu(\pi^\star)
+ 2\{r_V(n,\alpha)+\Delta_V\}
+ 2\lambda\{r_B(n,\alpha)+\Delta_B\},
\]
where
\begin{align*}
r_V(n,\alpha) &= M(1+2\hat\nu^{-2})\sqrt{\frac{\log(4|\Pi|/\alpha)}{2n}}, \\
r_B(n,\alpha) &= \hat\nu^{-2}\sqrt{\frac{\log(4|\Pi|/\alpha)}{2n}}.
\end{align*}
\end{corollary}

Corollary~\ref{cor:estimated} relaxes the known-propensity idealization in a deliberately conservative way. When $\hat e_i=e$ for all evaluation observations, the additive reconstruction terms vanish and the result has the same finite-class form as Theorem~\ref{thm:finite}, up to the looser envelope constant used for the reconstructed-propensity score. When propensities are reconstructed or estimated, the extra terms $\Delta_V$ and $\Delta_B$ separate value bias from burden miscalibration. Appendix~\ref{app:technical} makes the doubly robust product-error structure explicit. That is the form most relevant for the reconstructed MovieLens study, where logging probabilities are engineered rather than native production propensities.

The finite-class theorem is the main selection result for this paper because the empirical policy library is finite. Appendix~\ref{app:technical} briefly notes the extension to broader function classes and summarizes the relation between the formal assumptions and the empirical design.

\section{Experiments}

We evaluate four empirical questions: whether stagewise optimization can fail even in a transparent finite example; whether \CASP{} helps most when stage coupling is the dominant difficulty; how the method behaves when support becomes weak or the action space grows; and whether its support-aware objective induces a cleaner deployment-facing frontier than value-only or generic conservative selectors.

\subsection{Experimental Setup}

All results in this section use a shared simulator, nuisance-fit stack, policy library, and train-selection split. The simulation study contains five blocks: a minimal counterexample, a coupling sweep, a support-deficiency stress test, a large-action structured stress test, and a sample-size sweep. For the headline comparisons we fix \CASP{} at $\lambda=0.05$ and the generic conservative baseline at $\beta=0.50$. Appendix~\ref{app:technical} reports the full sweep grids, replication settings, and additional robustness tables.

\paragraph{Theory-to-experiment map.}
Block~1 illustrates the structural failure identified by Theorem~\ref{thm:stagewise}. Block~2 studies a coupling-dominant regime in which stage~1 mainly matters through downstream continuation value. Blocks~3 and~4 examine weaker support and larger structured action spaces, where the value--burden tradeoff becomes more pronounced. Block~5 varies sample size to assess finite-sample selection behavior. The MovieLens application examines the same value--burden logic under reconstructed logging probabilities, temporal holdout, and support diagnostics.

\subsection{Competitors and Metrics}

Each block compares the same policy families: a stagewise proxy learner, a plug-in reward baseline, a DR-value-only selector, a generic DR-LCB selector, the proposed \CASP{} selector, a two-stage OPL comparator following \citet{ma2020twostage}, and a downstream-aware candidate-generator comparator following \citet{wang2025candidategen}, with oracle included for interpretation where useful. All methods are evaluated using the same logged samples, nuisance-fit pipeline, policy library, train-selection split, and evaluation protocol. This common interface makes the comparison implementation-symmetric: differences in the reported results come from the selection rules rather than from different data splits, fitted nuisance models, or evaluation samples.

We report true policy value, oracle regret, support burden, and selected-policy mode frequency so that empirical conclusions are tied to both value and evaluability.

\subsection{Results}

\paragraph{Block 1: stagewise failure is real.}
\begin{table}[!htbp]
\centering
\small
\begin{tabular}{@{}lrrr@{}}
\toprule
Method & Value & Regret & Burden \\
\midrule
Stagewise & 0.2083 & 0.6417 & 2.00 \\
DR value only & 0.8500 & 0.0000 & 2.00 \\
CASP 0.05 & 0.8500 & 0.0000 & 2.00 \\
Ma-style OPL & 0.8500 & 0.0000 & 2.00 \\
Wang-style generator & 0.8500 & 0.0000 & 2.00 \\
Oracle & 0.8500 & 0.0000 & 2.00 \\
\bottomrule
\end{tabular}
\caption{Block 1 counterexample. The stagewise proxy selects the wrong generator, while every end-to-end learner recovers the optimum.}
\label{tab:block1-counterexample}
\end{table}

The counterexample gives the expected separation after the external baselines are added. Every end-to-end learner recovers the optimal policy, while the stagewise proxy selects the wrong generator. This is the empirical analogue of Theorem~\ref{thm:stagewise}: the failure mode is structural, not an artifact of comparing \CASP{} only against internal baselines.

\paragraph{Block 2: \CASP{} is strongest when coupling dominates.}
\begin{figure*}[t]
\centering
\begin{tikzpicture}
\begin{groupplot}[
    caspgroup,
    group style={group size=2 by 1, horizontal sep=1.55cm},
    width=0.44\textwidth,
    height=0.31\textwidth,
    xlabel={Coupling strength},
]
\nextgroupplot[
    ylabel={True policy value},
    legend to name=couplinglegend,
    legend columns=5,
    legend style={font=\scriptsize, draw=none, fill=none, /tikz/every even column/.append style={column sep=0.6em}},
]
\addplot[black, dashed, thick, mark=o, mark size=2pt] table[x=sweep_value,y=stagewise_value,col sep=comma]{figures/phase3_external_baselines/block2_coupling_key.csv};
\addlegendentry{Stagewise}
\addplot[orange!90!black, thick, mark=square*, mark size=2pt] table[x=sweep_value,y=dr_value_only_value,col sep=comma]{figures/phase3_external_baselines/block2_coupling_key.csv};
\addlegendentry{DR value only}
\addplot[green!50!black, thick, mark=triangle*, mark size=2pt] table[x=sweep_value,y=dr_lcb_value,col sep=comma]{figures/phase3_external_baselines/block2_coupling_key.csv};
\addlegendentry{DR-LCB 0.50}
\addplot[blue!70!black, very thick, mark=*, mark size=2.4pt] table[x=sweep_value,y=casp_value,col sep=comma]{figures/phase3_external_baselines/block2_coupling_key.csv};
\addlegendentry{CASP 0.05}
\addplot[red!75!black, thick, mark=diamond*, mark size=2.2pt] table[x=sweep_value,y=ma_style_value,col sep=comma]{figures/phase3_external_baselines/block2_coupling_key.csv};
\addlegendentry{Ma et al. OPL}
\nextgroupplot[
    ylabel={Support burden},
]
\addplot[black, dashed, thick, mark=o, mark size=2pt] table[x=sweep_value,y=stagewise_burden,col sep=comma]{figures/phase3_external_baselines/block2_coupling_key.csv};
\addplot[orange!90!black, thick, mark=square*, mark size=2pt] table[x=sweep_value,y=dr_value_only_burden,col sep=comma]{figures/phase3_external_baselines/block2_coupling_key.csv};
\addplot[green!50!black, thick, mark=triangle*, mark size=2pt] table[x=sweep_value,y=dr_lcb_burden,col sep=comma]{figures/phase3_external_baselines/block2_coupling_key.csv};
\addplot[blue!70!black, very thick, mark=*, mark size=2.4pt] table[x=sweep_value,y=casp_burden,col sep=comma]{figures/phase3_external_baselines/block2_coupling_key.csv};
\addplot[red!75!black, thick, mark=diamond*, mark size=2.2pt] table[x=sweep_value,y=ma_style_burden,col sep=comma]{figures/phase3_external_baselines/block2_coupling_key.csv};
\end{groupplot}
\path (group c1r1.south west) -- (group c2r1.south east)
node[midway,below=12pt] {\pgfplotslegendfromname{couplinglegend}};
\end{tikzpicture}
\caption{Block 2 coupling sweep. Left: true policy value. Right: support burden. As coupling strengthens, \CASP{} remains among the strongest learned selectors while maintaining lower burden than the value-only, generic conservative, and \citet{ma2020twostage} alternatives.}
\label{fig:coupling-headline}
\Description{Two-panel line plot for the coupling sweep. The left panel shows true policy value versus coupling strength for Stagewise, DR value only, DR-LCB 0.50, CASP 0.05, and \citet{ma2020twostage} OPL. CASP tracks the top learned value curve while Stagewise remains below the coupled methods. The right panel shows support burden versus coupling strength for the same methods; CASP stays below the value-only, DR-LCB, and \citet{ma2020twostage} comparators across the sweep.}
\end{figure*}
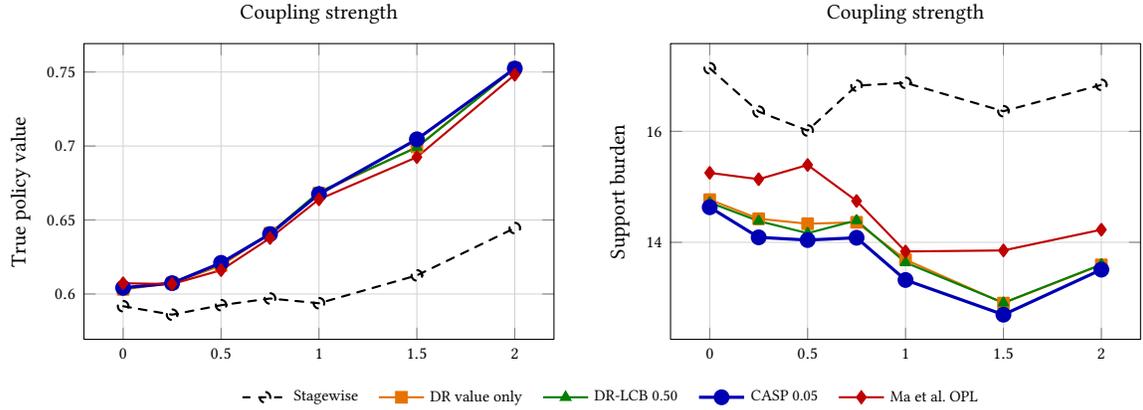

Figure~\ref{fig:coupling-headline} reports the coupling sweep. As coupling strengthens, the stagewise proxy falls behind the coupled learners. \CASP{} remains close to the top learned value curve while maintaining lower support burden than the DR-value-only, DR-LCB, and \citet{ma2020twostage} comparators. The \citet{wang2025candidategen} downstream-aware generator baseline is omitted from the figure for readability and is reported in the appendix block-level results.

\paragraph{Blocks 3--5: the harder regimes become a value--burden frontier.}
\begin{table}[!htbp]
\centering
\small
\begin{tabular}{@{}lrrr@{}}
\toprule
Method & Avg.\ value & Avg.\ regret & Avg.\ burden \\
\midrule
Stagewise & 0.6771 & 0.1524 & 38.90 \\
Plug-in & 0.7122 & 0.1173 & 44.98 \\
DR value only & 0.7176 & 0.1119 & 31.27 \\
DR-LCB 0.50 & 0.7175 & 0.1120 & 31.01 \\
CASP 0.05 & 0.7147 & 0.1148 & \textbf{26.77} \\
Ma-style OPL & \textbf{0.7216} & \textbf{0.1079} & 38.73 \\
Wang-style generator & 0.6960 & 0.1335 & 43.50 \\
\bottomrule
\end{tabular}
\caption{Cross-block averages over Blocks 2--5. Ma-style OPL is the strongest raw-value external comparator on average, while \CASP{} is the lowest-burden learned selector.}
\label{tab:crossblock-summary}
\end{table}

The remaining blocks are intentionally more mixed. In the support-stress, large-action, and sample-size sweeps, \citet{ma2020twostage} OPL is the strongest raw-value competitor on average. \CASP{} is therefore not a uniform raw-value winner. Its empirical advantage is different: it is the lowest-burden learned selector across the harder blocks, often by a substantial margin over both DR-only and \citet{ma2020twostage} alternatives. That makes the right claim deployment-facing rather than triumphalist.

\begin{figure*}[t]
\centering
\begin{tikzpicture}
\begin{groupplot}[
    caspgroup,
    group style={group size=2 by 1, horizontal sep=1.55cm},
    width=0.44\textwidth,
    height=0.31\textwidth,
    xlabel={Average support burden},
    every axis plot/.append style={only marks, mark size=2.6pt},
]
\nextgroupplot[
    ylabel={Average true policy value},
    legend to name=frontierlegend,
    legend columns=4,
    legend style={font=\scriptsize, draw=none, fill=none, /tikz/every even column/.append style={column sep=0.75em}},
]
\addplot[blue!70!black, mark=*,
    nodes near coords={\pgfplotspointmeta},
    point meta=explicit symbolic,
    every node near coord/.append style={font=\scriptsize, anchor=south west, xshift=2pt, yshift=1pt, text=blue!70!black, fill=white, fill opacity=0.75, text opacity=1, inner sep=1.2pt}
] table[x=casp_burden,y=casp_value,meta=block_label,col sep=comma]{figures/phase3_external_baselines/frontier_key.csv};
\addlegendentry{CASP 0.05}
\addplot[orange!90!black, mark=square*,
    nodes near coords={\pgfplotspointmeta},
    point meta=explicit symbolic,
    every node near coord/.append style={font=\scriptsize, anchor=north east, xshift=-2pt, yshift=-1pt, text=orange!90!black, fill=white, fill opacity=0.75, text opacity=1, inner sep=1.2pt}
] table[x=dr_value_only_burden,y=dr_value_only_value,meta=block_label,col sep=comma]{figures/phase3_external_baselines/frontier_key.csv};
\addlegendentry{DR value only}
\addplot[green!50!black, mark=triangle*,
    nodes near coords={\pgfplotspointmeta},
    point meta=explicit symbolic,
    every node near coord/.append style={font=\scriptsize, anchor=south east, xshift=-1pt, yshift=2pt, text=green!50!black, fill=white, fill opacity=0.75, text opacity=1, inner sep=1.2pt}
] table[x=dr_lcb_burden,y=dr_lcb_value,meta=block_label,col sep=comma]{figures/phase3_external_baselines/frontier_key.csv};
\addlegendentry{DR-LCB 0.50}
\addplot[red!75!black, mark=diamond*,
    nodes near coords={\pgfplotspointmeta},
    point meta=explicit symbolic,
    every node near coord/.append style={font=\scriptsize, anchor=north west, xshift=2pt, yshift=-2pt, text=red!75!black, fill=white, fill opacity=0.75, text opacity=1, inner sep=1.2pt}
] table[x=ma_style_burden,y=ma_style_value,meta=block_label,col sep=comma]{figures/phase3_external_baselines/frontier_key.csv};
\addlegendentry{Ma et al. OPL}
\nextgroupplot[
    ylabel={Average selection stability},
]
\addplot[blue!70!black, mark=*] table[x=casp_burden,y=casp_stability,col sep=comma]{figures/phase3_external_baselines/frontier_key.csv};
\addplot[orange!90!black, mark=square*] table[x=dr_value_only_burden,y=dr_value_only_stability,col sep=comma]{figures/phase3_external_baselines/frontier_key.csv};
\addplot[green!50!black, mark=triangle*] table[x=dr_lcb_burden,y=dr_lcb_stability,col sep=comma]{figures/phase3_external_baselines/frontier_key.csv};
\addplot[red!75!black, mark=diamond*] table[x=ma_style_burden,y=ma_style_stability,col sep=comma]{figures/phase3_external_baselines/frontier_key.csv};
\end{groupplot}
\path (group c1r1.south west) -- (group c2r1.south east)
node[midway,below=12pt] {\pgfplotslegendfromname{frontierlegend}};
\end{tikzpicture}
\caption{Cross-block frontier diagnostics over Blocks 2--5. Left: average value versus average support burden. Right: average selection stability versus average support burden. The labels B2--B5 denote the coupling, support-stress, large-action, and sample-size blocks. \CASP{} consistently occupies the lowest-burden region of the learned frontier, while \citet{ma2020twostage} OPL reaches higher raw value at materially larger burden.}
\label{fig:frontier-diagnostics}
\Description{Two-panel scatter plot summarizing Blocks 2 through 5. In the left panel, each method has one point per block showing average true value against average support burden. CASP lies at the lowest-burden end of the learned frontier, DR value only and DR-LCB lie in the middle, and \citet{ma2020twostage} OPL reaches the highest values at the largest burdens. In the right panel, the same blocks are plotted with average selection stability on the vertical axis; CASP remains moderately stable while \citet{ma2020twostage} OPL is less mode-concentrated and higher burden.}
\end{figure*}
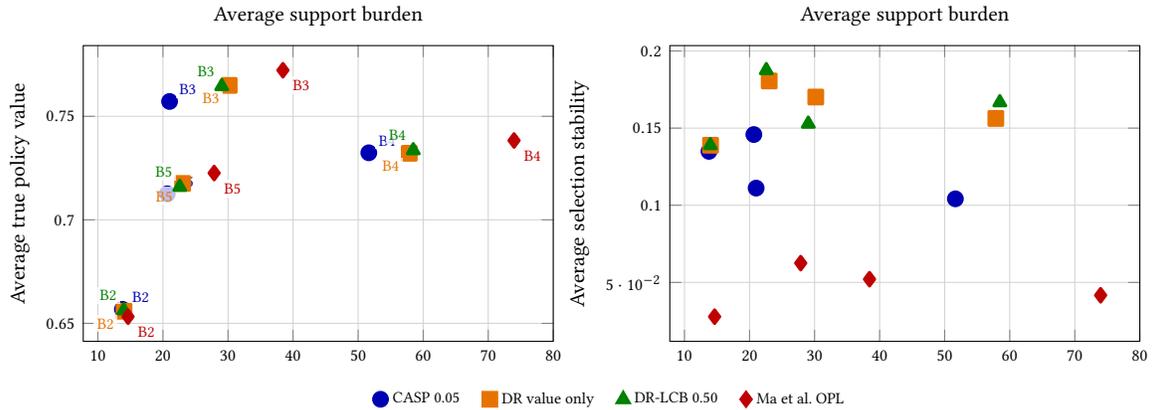

Figure~\ref{fig:frontier-diagnostics} summarizes the cross-block value--burden tradeoff. Across Blocks~2--5, \CASP{} occupies the lowest-burden region of the learned frontier, DR-value-only and DR-LCB lie in the middle, and \citet{ma2020twostage} OPL attains the highest average raw value at materially larger burden. The corresponding stability panel shows that the lower-burden \CASP{} region is not driven by degenerate or highly unstable selection.

A higher-sample-size check over $n\in\{600,1200,2400,4800\}$ gives the same qualitative pattern: \citet{ma2020twostage} OPL remains strongest on raw value, whereas \CASP{} remains the lowest-burden learned selector.

\paragraph{Takeaway.}
The simulation results support a focused empirical conclusion. \CASP{} is strongest when stage coupling is the dominant difficulty, and it remains the lowest-burden learned selector in harder regimes where raw value and evaluability are in tension.

\section{Semi-Synthetic MovieLens Application}

We evaluate \CASP{} on a reconstructed semi-synthetic \texttt{MovieLens 1M} pipeline. Because \texttt{MovieLens 1M} does not contain randomized exposure logs or native two-stage logging propensities, this application should not be interpreted as a production logged-bandit validation. Its purpose is narrower and more diagnostic: to test whether the proposed support-aware decision layer can expose policies that look attractive under offline value estimation but are not credible under the reconstructed generator--item support map. The application is therefore a controlled reconstruction study built from real user--item histories and an explicit synthetic two-stage logging layer.

\subsection{Dataset and Reconstruction Route}

The application uses the standard \texttt{MovieLens 1M} release, including the users, movies, and ratings files. Each eligible user--rating event after a warm-start history is treated as a request context. Contexts are constructed chronologically, and generator scores and feasible sets are computed using only prefix information available before the current event. The user history defines the request features; stage~1 selects a generator from a finite library, the selected generator exposes a fixed-size top-$L$ candidate set, and stage~2 selects a final item within that feasible set. The logged binary reward is derived from the observed rating label using the threshold $\text{rating}\ge 4$, and evaluation is based on temporal holdout rather than a random record split.

\subsection{Generator Library and Logging Design}

The application keeps stage~1 finite and interpretable. The generator library contains four members: a popularity-head retriever, a genre-match retriever, a collaborative-neighbor retriever, and a long-tail explorer. This library is small enough to keep the support structure legible and large enough to create meaningful exposure shifts. The reconstructed logger is defined explicitly at both stages through softmax sampling with minimum-mass floors. In the raw rating table, the positive-label rate under $\text{rating}\ge 4$ is about $0.575$.

To calibrate the reconstructed logger, we considered the stage~1 family
\[
\mu_1^{(\epsilon,\tau)}(g\mid x)
=(1-\epsilon)\,\mathrm{Softmax}_\tau(s(x))_g
\;+\;
\epsilon\,\mathrm{Unif}(\mathcal{G}(x)),
\]
over a small grid of $(\epsilon,\tau)$ values, motivated by randomized-logging, multi-policy OPE, and large-action smoothing methods \citep{li2011offline,saito2021obd,sachdeva2024policyconvolution}. Adjusting $\epsilon$ and $\tau$ alone did not give adequate generator diversity, so support construction was modified directly using temporal scanning, fallback candidate construction, hierarchical backoff for the collaborative-neighbor generator, and a capped singleton correction for generator~1. The final evaluation configuration uses $\epsilon=0.10$ and $\tau=1.00$, giving dominant stage~1 share $0.554$, minimum generator share $0.109$, strict $\ge 2$-support share $0.364$, mean generator-pair overlap $0.225$, and tail-item share $0.301$. Because MovieLens ratings are MNAR rather than true exposure logs, the resulting study is interpreted as reconstructed semi-synthetic evaluation rather than native logged-bandit evaluation \citep{wang2019mnar}.

The MovieLens study is not intended as a randomized-bandit validation of the formal assumptions. Instead, it examines whether the main objects of the theory can be implemented in a realistic recommender setting. The reconstructed pipeline therefore records induced feasible sets, stage~1 and stage~2 probabilities, support diagnostics, burden, stability, and policy deltas.

\subsection{Comparator Layer and Application Results}

The application uses the same primary comparator family as the simulation study: a stagewise learner, a plug-in reward baseline, a DR-value-only selector maximizing $\widehat V_{\DR}$, a generic DR-LCB selector, the proposed \CASP{} selector maximizing $\widehat J_\lambda$, the two-stage OPL comparator following \citet{ma2020twostage}, and the downstream-aware candidate-generator comparator following \citet{wang2025candidategen}. These external comparators are implemented within the same reconstructed \texttt{MovieLens 1M} pipeline as the other methods, using the same policy library, nuisance stack, logged samples, temporal splits, and evaluation protocol. Behavior, random, and a reconstructed-oracle diagnostic are included as reference rows. Additional ablations are reported in the appendix.

Table~\ref{tab:mlappmain} reports results over $20$ temporal replications. When policies are ranked solely by doubly robust value estimates, the \citet{wang2025candidategen} downstream-aware generator selector and the DR-value-only selector appear strongest. However, both operate at support burden on the order of $10^8$, indicating substantial mass assigned to choices outside reconstructed support under the $10^{-9}$ denominator floor. Appendix Table~\ref{tab:app-support-violation} shows floor-implied off-support mass of roughly $0.64$--$0.66$ for these high-DR selectors, whereas \CASP{} has essentially none. At $\lambda=0.05$, \CASP{} gives up some estimated value but reduces average burden from roughly $6.4$--$6.6\times 10^8$ to about $39.6$, with moderate instability across repeated splits. The \citet{ma2020twostage} external baseline remains close to DR-only on value but is highly unstable, selecting $13$ distinct policies across $20$ replications with modal frequency $0.10$. Consistent with the harder simulation blocks, the application supports interpreting \CASP{} as a support-aware conservative selector rather than a universal raw-value maximizer.

\begin{table}[!htbp]
\centering
\small
\begin{tabular}{lrrrr}
\toprule
Comparator & DR value & Burden & ESS & Max $w$ \\
\midrule
\CASP{} ($\lambda=0.05$) & 0.831 & 39.6 & 212.9 & 78.3 \\
DR value only & 0.888 & $6.58\times 10^8$ & 352.4 & 78.8 \\
DR-LCB ($\beta=0.50$) & 0.888 & $6.58\times 10^8$ & 352.4 & 78.8 \\
Ma-style OPL & 0.886 & $6.64\times 10^8$ & 405.0 & 78.8 \\
Wang-style generator & 0.898 & $6.41\times 10^8$ & 572.8 & 77.6 \\
\bottomrule
\end{tabular}
\caption{Main comparator summary for the reconstructed \texttt{MovieLens 1M} application. Higher DR value and ESS are better; lower burden and maximum importance weight are better.}
\label{tab:mlappmain}
\end{table}

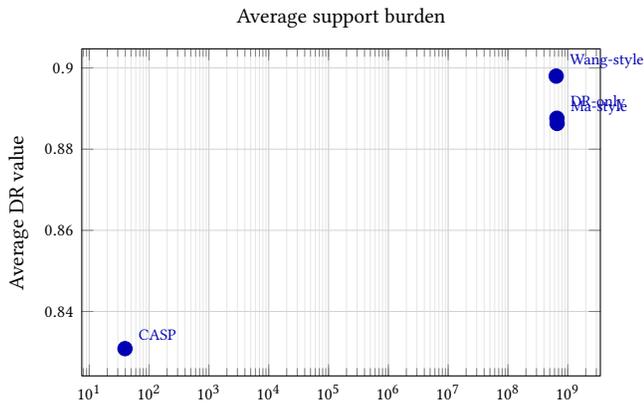
\begin{figure}[tbp]
\centering
\begin{tikzpicture}
\begin{axis}[
    caspaxis,
    width=\columnwidth,
    height=0.70\columnwidth,
    xlabel={Average support burden},
    ylabel={Average DR value},
    xmode=log,
    every axis plot/.append style={only marks, mark size=2.7pt},
]
\addplot[
    blue!70!black,
    mark=*,
    point meta=explicit symbolic,
    nodes near coords={\pgfplotspointmeta},
    every node near coord/.append style={
        font=\scriptsize,
        anchor=south west,
        xshift=2pt,
        text=blue!70!black
    }
] table[x=burden,y=value,meta=label,col sep=comma]{figures/application_assets/frontier_key.csv};
\end{axis}
\end{tikzpicture}
\caption{Application value--burden frontier for the main learned selectors. DR-LCB coincides with DR-value-only in this run and is omitted from the scatter for readability.}
\label{fig:app-frontier}
\Description{Scatter plot of average doubly robust value against average support burden for the application comparators CASP, DR-only, \citet{ma2020twostage}, and \citet{wang2025candidategen}. CASP lies at very low burden with lower DR value, while the other selectors cluster at much larger burdens and somewhat higher DR value.}
\end{figure}

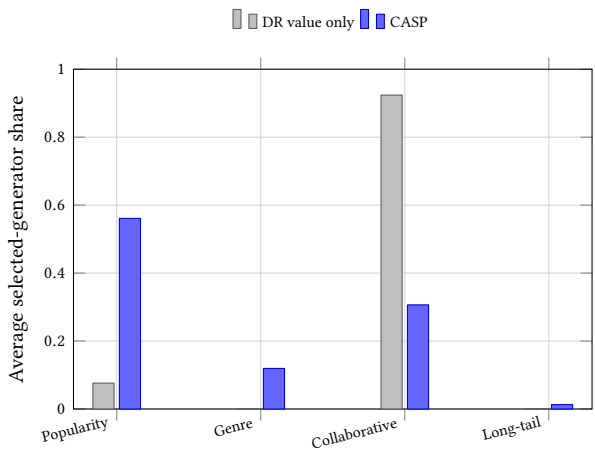
\begin{figure}[tbp]
\centering
\begin{tikzpicture}
\begin{axis}[
    caspaxis,
    ybar,
    bar width=8pt,
    width=\columnwidth,
    height=0.72\columnwidth,
    ymin=0,
    ymax=1,
    ylabel={Average selected-generator share},
    symbolic x coords={Popularity,Genre,Collaborative,LongTail},
    xtick=data,
    xticklabels={Popularity, Genre, Collaborative, Long-tail},
    x tick label style={font=\scriptsize, rotate=14, anchor=east},
    legend style={font=\scriptsize, draw=none, fill=none, at={(0.5,1.08)}, anchor=south, legend columns=2},
    grid=major,
]
\addplot[fill=gray!50, draw=gray!70!black] table[x=generator,y=dr_share,col sep=comma]{figures/application_assets/policy_delta_generator_shares.csv};
\addplot[fill=blue!60, draw=blue!80!black] table[x=generator,y=casp_share,col sep=comma]{figures/application_assets/policy_delta_generator_shares.csv};
\legend{DR value only, \CASP{}}
\end{axis}
\end{tikzpicture}
\caption{Generator-level policy delta between DR-value-only and \CASP{}. DR-value-only concentrates almost entirely on the collaborative-neighbor generator, while \CASP{} redistributes mass toward popularity-head and genre-match support.}
\label{fig:app-policy-delta}
\Description{Grouped bar chart comparing mean selected-generator shares under DR value only versus CASP across the four generators. DR value only is concentrated on the collaborative generator, while CASP has larger popularity and genre shares and a smaller collaborative share.}
\end{figure}

Figure~\ref{fig:app-frontier} summarizes the application value--burden tradeoff. \CASP{} occupies the low-burden region of the frontier, whereas the \citet{wang2025candidategen}, DR-only, and \citet{ma2020twostage} selectors lie at substantially larger burden. Figure~\ref{fig:app-policy-delta} shows that this burden reduction is associated with a stage~1 redistribution: DR-value-only selects generator~2 on about $92\%$ of evaluated contexts, whereas \CASP{} spreads mass across generators and gives average shares of approximately $(0.561, 0.120, 0.307, 0.013)$.

The large burden gap should be interpreted qualitatively rather than as a literal multiplicative improvement. The largest values arise when high-DR selectors place mass on generator--item pairs assigned zero support by the reconstructed logger; the denominator floor makes those violations finite and visible. For this reason, support burden is more informative here than generic ESS or maximum observed importance weight. ESS depends only on realized logged weights, whereas burden evaluates the target policy against the full reconstructed support map and therefore exposes off-support mass directly, even when those unsupported choices do not appear as extreme realized weights in the logged sample.

\paragraph{Robustness checks.}
The appendix reports sensitivity analyses for $\lambda$, effective sample size, maximum observed weight, off-support burden, and nearby MovieLens reconstruction settings. None of these analyses changes the main conclusion. In particular, moving from $\lambda=0$ to a positive penalty is sufficient to reject the high-DR off-support selector, while larger positive penalties preserve essentially the same low-burden policy family. In this application, the $\lambda$ curve therefore behaves more like a support-feasibility threshold than a smooth utility--risk frontier. Across the reported variants, \CASP{} remains in a low-burden range of roughly $25$--$41$, whereas DR-only and \citet{ma2020twostage} selection remain on the order of $6.6$--$7.4\times 10^8$.

\section{Related Work}

The closest architectural predecessor is the two-stage off-policy learning paper of \citet{ma2020twostage}, which shows that ignoring interaction between candidate generation and ranking can lead to suboptimal policies. More recent candidate-generator evaluation work studies downstream value under new action sets and makes explicit that retrieval quality should be judged in terms of the end recommendation task, not only through upstream retrieval metrics \citep{wang2025candidategen}. \CASP{} builds directly on these ideas, but shifts the center of gravity from two-stage interaction alone to stage-induced feasible support as the object that drives both offline learning and deployment-facing model selection. This distinction matters because modern large-scale recommenders often use multi-stage architectures in which candidate generation first narrows a large item universe before a more expensive ranker scores the exposed candidates \citep{covington2016youtube,chen2019topk}. Classical recommender objectives such as implicit-feedback matrix factorization and Bayesian personalized ranking are useful for representation learning and ranking \citep{hu2008implicit,rendle2009bpr}, but they do not directly address whether a new generator--ranker policy is credible under the logged support created by a two-stage pipeline.

At a broader level, the paper inherits the main estimators and learning principles of contextual-bandit OPE and OPL, especially IPS, self-normalized estimation, doubly robust estimation, shrinkage, and counterfactual risk minimization \citep{li2011offline,dudik2014doublyrobust,swaminathan2015crm,swaminathan2015jmlr,swaminathan2015snips,wang2017optimal,su2020shrinkage}. Related work in counterfactual recommendation and learning to rank emphasizes that logged recommendation data are biased by the policy that generated exposure, rather than being an unbiased sample of user preferences \citep{bottou2013counterfactual,schnabel2016recommendations,joachims2017unbiasedltr,agarwal2018cltr,joachims2018deep,gilotte2018offline}. These ideas motivate the use of logged propensities and counterfactual estimators in our work. However, standard contextual-bandit and counterfactual-ranking formulations often treat the feasible action set as fixed once the context is observed. Two-stage recommenders violate that simplification because the first-stage generator changes the support on which the second-stage ranker operates.

This support issue connects the paper to large-action and structured-action evaluation literatures. Slate OPE shows that naive importance weighting is often untenable for structured recommendation outputs, especially when rewards depend on multiple displayed items or sequential user responses \citep{swaminathan2017slate,vlassis2021slatecv,mcinerney2020sequentialslate,ie2019slateq}. Large-action OPE shows that weak support and action-side generalization are central obstacles even before one adds architectural coupling \citep{sachdeva2020deficientsupport,saito2022embeddings,saito2023offcem,peng2023actiongrouping,sachdeva2024policyconvolution}. Reproducible logged-bandit resources and robustness studies further show that OPE conclusions can be sensitive to logging design, estimator choice, and support quality \citep{saito2021obd,openbanditpipeline,saito2021robustness}. The present paper absorbs those lessons but asks a different question: how should one select a coupled two-stage policy when support deficiency is created upstream by the candidate generator itself?

Finally, the paper sits alongside pessimistic and conservative offline learning in recommendation. Pessimistic reward modeling and multivariate policy learning both show that deployment-facing offline decisions should guard against unsupported optimism \citep{jeunen2021pessimistic,jeunen2022pessimisticjournal,jeunen2024multivariate}. Recent surveys on causal recommendation and offline reinforcement learning for recommender systems also emphasize that logged recommendation data involve exposure bias, confounding, feedback loops, distribution shift, and limited support \citep{luo2024causalsurvey,gao2024causalrecsurvey,chen2024offlinerlrecsys,levine2020offlinerlsurvey}. Our contribution is to make that conservative logic explicitly two-stage and support-structural: the penalty is tied to how the first-stage generator shapes downstream evaluability. Thus, \CASP{} converts offline evaluation from a single value-ranking problem into a value--support selection problem for two-stage recommender deployment.

\section{Discussion and Limitations}

This paper studies a deliberately narrow setting: stage~1 is a finite generator choice, stage~2 is a single-item choice from an induced feasible set, and learning is based on a coupled pessimistic objective. Within this scope, the theory provides a population guarantee, a finite-class selection result, and a reconstructed-propensity extension. It does not address broader surrogate-consistency questions, fully unknown-propensity orthogonalization with shared cross-fitted nuisances, or slate-valued second-stage actions. The MovieLens study is reconstructed and semi-synthetic rather than a production randomized-logging evaluation, so its role is illustrative rather than definitive.

\paragraph{Calibration limits in the application.}
The MovieLens $\lambda$ sensitivity is best interpreted as threshold behavior rather than as a smoothly calibrated risk--utility frontier. In the reported application, moving from $\lambda=0$ to a positive penalty is sufficient to reject a high-DR off-support selector, whereas larger positive penalties mostly preserve the same low-burden policy family. A smoother frontier would likely require a richer candidate policy library, more graded overlap, or a different burden-normalization design.

\paragraph{When \CASP{} is most useful.}
\CASP{} is most useful when high estimated value and reliable support are genuinely in tension. If the logged system already provides strong overlap for all credible policies, a value-only selector may be adequate. Conversely, if the decision maker is willing to prioritize raw value despite weak evaluability, the burden penalty may be too conservative. The method also depends on the quality of the reconstructed logging model, since the burden diagnostic is calibrated through that model.

Natural extensions include broader function-class theory, unknown-propensity settings, slate-valued second-stage decisions, and richer stage~1 parameterizations. These directions are important, but they are not required for the support-aware contribution developed here.

\section*{Code Availability}
The code for \CASP{}, including the simulation pipeline, the semi-synthetic MovieLens application, experiment configurations, and paper-asset generation scripts, is publicly available at
\url{https://github.com/nilson01/CASP}.

\bibliographystyle{ACM-Reference-Format}
\bibliography{refs}

\clearpage
\appendix
\onecolumn

\section{Proofs and Additional Results}
\label{app:technical}

This appendix contains the full proofs, supporting lemmas, and additional empirical results referenced in the main text.

\subsection{Proof of Proposition~\ref{prop:dynamic}}

\begin{proof}
Fix a stage-1 rule $\pi_1$. For any realized context $x$ and action $a_1$, the stage-2 contribution
\[
\sum_{a_2\in S(x,a_1)}\pi_2(a_2\mid x,a_1)q(x,a_1,a_2)
\]
is linear in the probability vector $\pi_2(\cdot\mid x,a_1)$ over the simplex on $S(x,a_1)$. A linear functional over a simplex is maximized at an extreme point, so for each $(x,a_1)$ the optimal stage-2 rule places all mass on an action attaining
\[
m^\star(x,a_1)=\max_{a_2\in S(x,a_1)} q(x,a_1,a_2).
\]
Substituting this pointwise optimum into \eqref{eq:value} gives
\[
\sup_{\pi_2}V(\pi_1,\pi_2)
=
\E\Bigg[
\sum_{a_1\in\cA_1}\pi_1(a_1\mid X)m^\star(X,a_1)
\Bigg].
\]
The displayed quantity is again linear in the stage-1 probability vector $\pi_1(\cdot\mid x)$ for each fixed $x$, so a value-optimal stage-1 rule may be chosen by putting all mass on an action in $\argmax_{a_1} m^\star(x,a_1)$.
\end{proof}

\subsection{Proof of Theorem~\ref{thm:stagewise}}

\begin{proof}
Take a degenerate context distribution with $X\equiv x_0$ almost surely. Choose two distinct downstream items $u,v\in\cA_2$. Define singleton feasible sets
\[
S(x_0,a)=\{u\}\quad \text{for all } a\in A_g^\star(x_0),
\qquad
S(x_0,b)=\{v\}.
\]
Assign any nonempty feasible set to the remaining stage-1 actions, if any.
Set the reward regression to
\[
q(x_0,a,u)=0\quad \text{for all } a\in A_g^\star(x_0),
\qquad
q(x_0,b,v)=M.
\]
Assign reward regression value $0$ to all remaining feasible triples, and choose a logging policy with positive probability on all feasible pairs used in the construction. Because each displayed feasible set is a singleton, stage 2 is forced once one of those stage-1 actions is chosen, regardless of the stage-2 score $h$. Every continuation-blind separable rule puts first-stage mass only on $A_g^\star(x_0)$ by definition of the stage-1 proxy maximizer set, and therefore has value $0$.

By Proposition~\ref{prop:dynamic}, however, the correct stage-1 target is the downstream continuation value. In this construction,
\[
m^\star(x_0,a)=0\quad \text{for all } a\in A_g^\star(x_0),
\qquad
m^\star(x_0,b)=M,
\]
so any end-to-end value-optimal policy selects $b$ at $x_0$. The resulting value gap is exactly $M$. The construction clearly satisfies finiteness, consistency, ignorability, and support compatibility, which proves the claim.
\end{proof}

\subsection{Proof of Proposition~\ref{prop:secondmoment}}

\begin{proof}
Condition on $X=x$. Using the logged policy to expand the conditional expectation,
\begin{align*}
\E[w_\pi(X,A_1,A_2)^2\mid X=x]
&=
\sum_{a_1\in\cA_1}\sum_{a_2\in S(x,a_1)}
\mu_1(a_1\mid x)\mu_2(a_2\mid x,a_1)
\Bigg(
\frac{\pi_1(a_1\mid x)\pi_2(a_2\mid x,a_1)}
     {\mu_1(a_1\mid x)\mu_2(a_2\mid x,a_1)}
\Bigg)^2 \\
&=
\sum_{a_1\in\cA_1}\sum_{a_2\in S(x,a_1)}
\frac{\pi_1(a_1\mid x)^2\pi_2(a_2\mid x,a_1)^2}
     {\mu_1(a_1\mid x)\mu_2(a_2\mid x,a_1)} \\
&= B_\mu(\pi;x).
\end{align*}
Taking expectation over $X$ gives $\E[w_\pi(X,A_1,A_2)^2]=B_\mu(\pi)$.

For the IPS estimator,
\[
\widehat V_{\IPS}(\pi)=\frac{1}{n}\sum_{i=1}^n w_\pi(X_i,A_{1i},A_{2i})Y_i.
\]
The observations are i.i.d., so
\[
\operatorname{Var}\{\widehat V_{\IPS}(\pi)\}
=
\frac{1}{n}\operatorname{Var}\{w_\pi(X,A_1,A_2)Y\}
\le
\frac{1}{n}\E[w_\pi(X,A_1,A_2)^2Y^2].
\]
Since $0\le Y\le M$,
\[
\E[w_\pi(X,A_1,A_2)^2Y^2]
\le
M^2\E[w_\pi(X,A_1,A_2)^2]
=
M^2B_\mu(\pi).
\]
Therefore
\[
\operatorname{Var}\{\widehat V_{\IPS}(\pi)\}
\le
\frac{M^2}{n}B_\mu(\pi).
\]
\end{proof}

\subsection{Proof of Theorem~\ref{thm:popregret}}

\begin{proof}
Because $\pi_\lambda$ maximizes $J_\lambda(\pi)=V(\pi)-\lambda B_\mu(\pi)$ over $\Pi$,
\[
J_\lambda(\pi_\lambda)\ge J_\lambda(\pi^\star).
\]
Expanding both sides gives
\[
V(\pi_\lambda)-\lambda B_\mu(\pi_\lambda)
\ge
V(\pi^\star)-\lambda B_\mu(\pi^\star).
\]
Rearranging,
\[
V(\pi^\star)-V(\pi_\lambda)
\le
\lambda\bigl(B_\mu(\pi^\star)-B_\mu(\pi_\lambda)\bigr).
\]
Since $B_\mu(\pi_\lambda)\ge 0$, we obtain
\[
V(\pi^\star)-V(\pi_\lambda)\le \lambda B_\mu(\pi^\star),
\]
as claimed.
\end{proof}

\subsection{Auxiliary Lemmas for Conservative Selection}

\begin{lemma}[Conditional unbiasedness of the doubly robust score]
\label{lem:drunbiasedappendix}
Assume the out-of-sample scoring convention used in the definition of $\widehat V_{\DR}$, so that each fitted reward model used to score an observation is independent of that observation. Then for every fixed policy $\pi$,
\[
\E[\psi_\pi(O;\hat q)\mid \hat q]=V(\pi).
\]
Consequently, $\E[\widehat V_{\DR}(\pi)]=V(\pi)$.
\end{lemma}

\begin{proof}
Condition on $\hat q$ and on a generic evaluation observation $O=(X,A_1,A_2,Y)$ that is independent of the fit used to score it. By definition of \eqref{eq:drscore},
\begin{align*}
\E[\psi_\pi(O;\hat q)\mid \hat q]
&=
\E[m_{\hat q}(X;\pi)\mid \hat q] \\
&\quad+
\E\Bigg[
\frac{\pi_1(A_1\mid X)\pi_2(A_2\mid X,A_1)}
     {\mu_1(A_1\mid X)\mu_2(A_2\mid X,A_1)}
\bigl(Y-\hat q(X,A_1,A_2)\bigr)
\Bigm|\hat q
\Bigg].
\end{align*}
Now condition on $X=x$. The second term becomes
\begin{align*}
&\sum_{a_1\in\cA_1}\sum_{a_2\in S(x,a_1)}
\mu_1(a_1\mid x)\mu_2(a_2\mid x,a_1)
\frac{\pi_1(a_1\mid x)\pi_2(a_2\mid x,a_1)}
     {\mu_1(a_1\mid x)\mu_2(a_2\mid x,a_1)}
\bigl(q(x,a_1,a_2)-\hat q(x,a_1,a_2)\bigr) \\
&=
\sum_{a_1\in\cA_1}\sum_{a_2\in S(x,a_1)}
\pi_1(a_1\mid x)\pi_2(a_2\mid x,a_1)
\bigl(q(x,a_1,a_2)-\hat q(x,a_1,a_2)\bigr).
\end{align*}
Adding the plug-in term $m_{\hat q}(x;\pi)$ cancels the $\hat q$ terms and leaves
\[
\sum_{a_1\in\cA_1}\sum_{a_2\in S(x,a_1)}
\pi_1(a_1\mid x)\pi_2(a_2\mid x,a_1)q(x,a_1,a_2).
\]
Taking expectation over $X$ gives $V(\pi)$ by \eqref{eq:value}. Averaging over the independent evaluation observations gives $\E[\widehat V_{\DR}(\pi)]=V(\pi)$.
\end{proof}

\begin{lemma}[Uniform-selection reduction]
\label{lem:oracleappendix}
Let
\[
J_\lambda(\pi)=V(\pi)-\lambda B_\mu(\pi),
\qquad
\widehat J_\lambda(\pi)=\widehat V_{\DR}(\pi)-\lambda\widehat B_\mu(\pi),
\]
and define
\[
\pi_\lambda\in\argmax_{\pi\in\Pi}J_\lambda(\pi),
\qquad
\widehat\pi_\lambda\in\argmax_{\pi\in\Pi}\widehat J_\lambda(\pi).
\]
If
\[
\sup_{\pi\in\Pi}|\widehat V_{\DR}(\pi)-V(\pi)|\le \varepsilon_V
\qquad\text{and}\qquad
\sup_{\pi\in\Pi}|\widehat B_\mu(\pi)-B_\mu(\pi)|\le \varepsilon_B,
\]
then
\[
V(\pi^\star)-V(\widehat\pi_\lambda)
\le
\lambda B_\mu(\pi^\star)+2\varepsilon_V+2\lambda\varepsilon_B.
\]
\end{lemma}

\begin{proof}
On the displayed event, every candidate policy satisfies
\[
\widehat J_\lambda(\pi)\le J_\lambda(\pi)+\varepsilon_V+\lambda\varepsilon_B
\qquad\text{and}\qquad
J_\lambda(\pi)\le \widehat J_\lambda(\pi)+\varepsilon_V+\lambda\varepsilon_B.
\]
Because $\widehat\pi_\lambda$ maximizes $\widehat J_\lambda$,
\begin{align*}
J_\lambda(\pi_\lambda)
&\le \widehat J_\lambda(\pi_\lambda)+\varepsilon_V+\lambda\varepsilon_B \\
&\le \widehat J_\lambda(\widehat\pi_\lambda)+\varepsilon_V+\lambda\varepsilon_B \\
&\le J_\lambda(\widehat\pi_\lambda)+2\varepsilon_V+2\lambda\varepsilon_B.
\end{align*}
Now write
\begin{align*}
V(\pi^\star)-V(\widehat\pi_\lambda)
&=
\bigl(V(\pi^\star)-J_\lambda(\pi^\star)\bigr)
+\bigl(J_\lambda(\pi^\star)-J_\lambda(\widehat\pi_\lambda)\bigr)
+\lambda B_\mu(\widehat\pi_\lambda) \\
&\le
\lambda B_\mu(\pi^\star)+J_\lambda(\pi_\lambda)-J_\lambda(\widehat\pi_\lambda),
\end{align*}
where we used $J_\lambda(\pi_\lambda)\ge J_\lambda(\pi^\star)$ and $\lambda B_\mu(\widehat\pi_\lambda)\ge 0$. Combining the two displays proves the claim.
\end{proof}

\subsection{Proof of Theorem~\ref{thm:finite}}

\begin{proof}
Fix $\pi\in\Pi$. By Lemma~\ref{lem:drunbiasedappendix},
\[
\E[\widehat V_{\DR}(\pi)]=V(\pi).
\]
Under the lower bound $\mu_1,\mu_2\ge \nu$ on the relevant support and the reward bound $0\le Y\le M$, each doubly robust score obeys
\[
|m_{\hat q}(X;\pi)|\le M
\qquad\text{and}\qquad
\left|
\frac{\pi_1(A_1\mid X)\pi_2(A_2\mid X,A_1)}
     {\mu_1(A_1\mid X)\mu_2(A_2\mid X,A_1)}
\bigl(Y-\hat q(X,A_1,A_2)\bigr)
\right|
\le M\nu^{-2},
\]
so
\[
|\psi_\pi(O;\hat q)|\le M(1+\nu^{-2})=:C_V.
\]
Hoeffding's inequality therefore gives
\[
\Pbb\Bigl(|\widehat V_{\DR}(\pi)-V(\pi)|>t\Bigr)
\le 2\exp\Bigl(-\frac{2nt^2}{C_V^2}\Bigr).
\]
Applying a union bound over the finite class $\Pi$ and setting the right-hand side to $\alpha/2$ gives
\[
\sup_{\pi\in\Pi}\bigl|\widehat V_{\DR}(\pi)-V(\pi)\bigr|
\le
M(1+\nu^{-2})\sqrt{\frac{\log(4|\Pi|/\alpha)}{2n}}
\]
with probability at least $1-\alpha/2$.

Next fix $\pi\in\Pi$ and consider the bounded random variable $B_\mu(\pi;X)$. Since each squared-probability sum is at most one,
\begin{align*}
B_\mu(\pi;X)
&=
\sum_{a_1\in\cA_1}\sum_{a_2\in S(X,a_1)}
\frac{\pi_1(a_1\mid X)^2\pi_2(a_2\mid X,a_1)^2}
     {\mu_1(a_1\mid X)\mu_2(a_2\mid X,a_1)} \\
&\le
\nu^{-2}
\sum_{a_1\in\cA_1}\pi_1(a_1\mid X)^2
\sum_{a_2\in S(X,a_1)}\pi_2(a_2\mid X,a_1)^2 \\
&\le \nu^{-2}.
\end{align*}
Another Hoeffding-plus-union-bound argument gives
\[
\sup_{\pi\in\Pi}\bigl|\widehat B_\mu(\pi)-B_\mu(\pi)\bigr|
\le
\nu^{-2}\sqrt{\frac{\log(4|\Pi|/\alpha)}{2n}}
\]
with probability at least $1-\alpha/2$. Intersecting the two events gives the pair of uniform bounds in the theorem statement with probability at least $1-\alpha$.

Finally apply Lemma~\ref{lem:oracleappendix} with
\[
\varepsilon_V=M(1+\nu^{-2})\sqrt{\frac{\log(4|\Pi|/\alpha)}{2n}}
\qquad\text{and}\qquad
\varepsilon_B=\nu^{-2}\sqrt{\frac{\log(4|\Pi|/\alpha)}{2n}},
\]
which gives the advertised value bound for $\widehat\pi_\lambda$.
\end{proof}

\subsection{Estimated Propensities and Nuisance Error}

Let
\[
e(x,a_1,a_2)=\mu_1(a_1\mid x)\mu_2(a_2\mid x,a_1),
\]
and let $\hat e_i$ be an out-of-sample estimate or reconstructed analogue of the same joint propensity for evaluation observation $i$. Define the estimated-propensity score
\[
\hat\psi_{\pi,i}
=
m_{\hat q_i}(X_i;\pi)
+
\frac{\pi_1(A_{1i}\mid X_i)\pi_2(A_{2i}\mid X_i,A_{1i})}
     {\hat e_i(X_i,A_{1i},A_{2i})}
\bigl(Y_i-\hat q_i(X_i,A_{1i},A_{2i})\bigr),
\]
the corresponding value estimator
\[
\widehat V_{\DR,\hat e}(\pi)=\frac{1}{n}\sum_{i=1}^n \hat\psi_{\pi,i},
\]
and the corresponding empirical burden
\[
\widehat B_{\hat e}(\pi)=\frac{1}{n}\sum_{i=1}^n
\sum_{a_1\in\cA_1}\sum_{a_2\in S(X_i,a_1)}
\frac{\pi_1(a_1\mid X_i)^2\pi_2(a_2\mid X_i,a_1)^2}
     {\hat e_i(X_i,a_1,a_2)}.
\]
For fold-specific propensities, define the nuisance-conditional population burden average
\begin{align*}
B_{\hat e,n}(\pi)
\;=\;&
\frac{1}{n}\sum_{i=1}^n
\E\Bigg[
\sum_{a_1\in\cA_1}\sum_{a_2\in S(X,a_1)}
\frac{\pi_1(a_1\mid X)^2\pi_2(a_2\mid X,a_1)^2}
     {\hat e_i(X,a_1,a_2)}
\;\Bigg|\;\hat e_i
\Bigg].
\end{align*}

\begin{corollary}[Finite-class guarantee with estimated propensities]
\label{cor:estimatedappendix}
Suppose $\Pi$ is finite, each evaluation observation is scored with out-of-sample $\hat q_i$ and $\hat e_i$, $0\le \hat q_i\le M$, and $\hat e_i\ge \hat\nu^2$ on the relevant support. Define
\[
\Delta_V
=
\sup_{\pi\in\Pi}
\left|
\E[\widehat V_{\DR,\hat e}(\pi)\mid \{\hat q_i,\hat e_i\}_{i=1}^n]-V(\pi)
\right|,
\]
\[
\Delta_B
=
\sup_{\pi\in\Pi}\bigl|B_{\hat e,n}(\pi)-B_\mu(\pi)\bigr|.
\]
Then, conditional on the nuisance fits, with probability at least $1-\alpha$,
\[
\sup_{\pi\in\Pi}
\bigl|\widehat V_{\DR,\hat e}(\pi)-V(\pi)\bigr|
\le
r_V(n,\alpha)+\Delta_V,
\]
\[
\sup_{\pi\in\Pi}
\bigl|\widehat B_{\hat e}(\pi)-B_\mu(\pi)\bigr|
\le
r_B(n,\alpha)+\Delta_B,
\]
where
\begin{align*}
r_V(n,\alpha) &= M(1+2\hat\nu^{-2})\sqrt{\frac{\log(4|\Pi|/\alpha)}{2n}}, \\
r_B(n,\alpha) &= \hat\nu^{-2}\sqrt{\frac{\log(4|\Pi|/\alpha)}{2n}}.
\end{align*}
Consequently, if
\[
\widehat\pi_{\lambda,\hat e}
\in
\argmax_{\pi\in\Pi}
\left\{
\widehat V_{\DR,\hat e}(\pi)-\lambda\widehat B_{\hat e}(\pi)
\right\},
\]
then
\[
V(\pi^\star)-V(\widehat\pi_{\lambda,\hat e})
\le
\lambda B_\mu(\pi^\star)
+ 2\{r_V(n,\alpha)+\Delta_V\}
+ 2\lambda\{r_B(n,\alpha)+\Delta_B\}
\]
with the same probability.
\end{corollary}

\begin{proof}
Fix $\pi\in\Pi$ and condition on the fitted nuisances. By construction,
\[
\E[\widehat V_{\DR,\hat e}(\pi)\mid \{\hat q_i,\hat e_i\}_{i=1}^n]
=
V(\pi)+\mathrm{Bias}_{\hat e,\hat q}(\pi),
\]
where the conditional bias term is absorbed into $\Delta_V$ by definition. Under $0\le \hat q_i\le M$ and $\hat e_i\ge \hat\nu^2$,
\[
\left|
\frac{\pi_1(A_{1i}\mid X_i)\pi_2(A_{2i}\mid X_i,A_{1i})}
     {\hat e_i(X_i,A_{1i},A_{2i})}
\bigl(Y_i-\hat q_i(X_i,A_{1i},A_{2i})\bigr)
\right|
\le M\hat\nu^{-2},
\]
so
\[
|\hat\psi_{\pi,i}|\le M(1+2\hat\nu^{-2}).
\]
Hoeffding's inequality and a union bound over $\Pi$ therefore give
\[
\sup_{\pi\in\Pi}
\left|
\widehat V_{\DR,\hat e}(\pi)
-\E[\widehat V_{\DR,\hat e}(\pi)\mid \{\hat q_i,\hat e_i\}_{i=1}^n]
\right|
\le
r_V(n,\alpha)
\]
with probability at least $1-\alpha/2$. Combining this with the definition of $\Delta_V$ gives the first display.

For the burden term, each summand is nonnegative and bounded above by $\hat\nu^{-2}$ because the squared policy probabilities sum to at most one. Another Hoeffding-plus-union-bound argument gives
\[
\sup_{\pi\in\Pi}
\bigl|
\widehat B_{\hat e}(\pi)-B_{\hat e,n}(\pi)
\bigr|
\le
r_B(n,\alpha)
\]
with probability at least $1-\alpha/2$. Combining this with the definition of $\Delta_B$ proves the second display. Intersect the two events and apply Lemma~\ref{lem:oracleappendix} with $\varepsilon_V=r_V(n,\alpha)+\Delta_V$ and $\varepsilon_B=r_B(n,\alpha)+\Delta_B$.
\end{proof}

\begin{remark}[Product-error structure behind $\Delta_V$]
The value-bias term $\Delta_V$ vanishes if either the propensity model is exact or the reward nuisance is exact. More generally, the same decomposition used in standard doubly robust arguments gives a product-error dependence of the form
\[
\Delta_V
\le
\hat\nu^{-2}
\|\hat q-q\|_\infty
\|\hat e-e\|_\infty
\]
whenever the two sup-norm errors are finite. This is why Corollary~\ref{cor:estimatedappendix} is the right scoped extension: it keeps the practical estimated-propensity story without claiming that the full unknown-propensity problem is solved.
\end{remark}

\subsection{Function-Class Bridge}

\begin{corollary}[Function-class bridge for \CASP{} selection]
\label{cor:functionbridge}
Let $\Pi$ be an arbitrary policy class, not necessarily finite, and condition on fixed out-of-sample nuisance fits. Define the score class
\[
\mathcal{F}
=
\left\{
O\mapsto \hat\psi_\pi(O):\pi\in\Pi
\right\},
\]
and the burden class
\[
\mathcal{G}
=
\left\{
X\mapsto
\sum_{a_1\in\cA_1}\sum_{a_2\in S(X,a_1)}
\frac{\pi_1(a_1\mid X)^2\pi_2(a_2\mid X,a_1)^2}
     {\hat e(X,a_1,a_2)}
:
\pi\in\Pi
\right\}.
\]
Assume $\mathcal{F}$ and $\mathcal{G}$ admit envelopes $C_V$ and $C_B$. Then, with probability at least $1-\alpha$,
\[
\sup_{\pi\in\Pi}
\bigl|
\widehat V_{\DR,\hat e}(\pi)-V(\pi)
\bigr|
\le
2\mathfrak{R}_n(\mathcal{F})
+
C_V\sqrt{\frac{\log(4/\alpha)}{2n}}
+
\Delta_V,
\]
\[
\sup_{\pi\in\Pi}
\bigl|
\widehat B_{\hat e}(\pi)-B_\mu(\pi)
\bigr|
\le
2\mathfrak{R}_n(\mathcal{G})
+
C_B\sqrt{\frac{\log(4/\alpha)}{2n}}
+
\Delta_B.
\]
Consequently, the same selection reduction as in Lemma~\ref{lem:oracleappendix} gives a \CASP{} regret bound with the finite-class $\log|\Pi|/n$ terms replaced by the two empirical complexity terms above. If a finite policy library is viewed as a sieve approximation to a richer class, the resulting bound also gains the usual approximation error term.
\end{corollary}

\begin{proof}
Conditional on the nuisance fits, the score and burden processes are bounded empirical processes indexed by $\pi$. Symmetrization and empirical-process concentration therefore give the two displayed uniform deviations, with the conditional bias and burden-miscalibration terms still entering through $\Delta_V$ and $\Delta_B$. The final regret statement then follows by applying Lemma~\ref{lem:oracleappendix} exactly as in the finite-class proof.
\end{proof}

\subsection{Assumptions and Empirical Alignment}
\label{app:audit}

The formal assumptions of the theory hold by construction in simulation and only approximately in the reconstructed MovieLens study. In simulation, actions, feasible sets, rewards, and propensities are generated directly from the data-generating process. In the MovieLens application, feasibility and propensities are reconstructed from chronological prefixes and explicit logging rules, so the results should be interpreted as semi-synthetic evaluation under observational ratings.

This distinction motivates the joint reporting of value, burden, and support diagnostics in the application. Misspecified reconstructed propensities can bias both value and burden, and weak support can lead to unstable policy comparisons.

\subsection{Experimental Details}
\label{app:expdetails}

This appendix contains the implementation details needed to interpret and reproduce the empirical results. These include the simulation data-generating mechanisms, the shared generator and policy libraries, feasible-set construction rules, tuning grids, tie-breaking conventions, sample-splitting rules, replication settings, and the additional robustness analyses omitted from the main text.

Across all reported comparisons, competitors share the same logged samples, nuisance-fit pipeline, policy library, and evaluation protocol. This common design ensures that observed differences are attributable to the learning rule rather than to implementation asymmetries.

\subsection{Full Per-Block Comparator Tables}

This subsection reports the full comparator tables for each simulation block using the same comparison layer as the main results: stagewise, plug-in, DR value only, DR-LCB, \CASP{}, \citet{ma2020twostage} OPL, \citet{wang2025candidategen} generator selection, and oracle.

\begin{table}[!htbp]
\centering
\small
\begin{tabular}{@{}lrrrr@{}}
\toprule
Method & Value & Regret & Burden & Stability \\
\midrule
Stagewise & 0.2083 & 0.6417 & 2.00 & 0.000 \\
Plug-in & 0.8500 & 0.0000 & 2.00 & 0.000 \\
DR value only & 0.8500 & 0.0000 & 2.00 & 1.000 \\
DR-LCB 0.50 & 0.8500 & 0.0000 & 2.00 & 1.000 \\
CASP 0.05 & 0.8500 & 0.0000 & 2.00 & 1.000 \\
Ma-style OPL & 0.8500 & 0.0000 & 2.00 & 1.000 \\
Wang-style generator & 0.8500 & 0.0000 & 2.00 & 1.000 \\
Oracle & 0.8500 & 0.0000 & 2.00 & 0.000 \\
\bottomrule
\end{tabular}
\caption{Block 1 full comparator table. For Blocks 2--5 the entries are sweep averages; Block 1 is reported directly.}
\label{tab:app-block1-full}
\end{table}

\begin{table}[!htbp]
\centering
\small
\begin{tabular}{@{}lrrrr@{}}
\toprule
Method & Value & Regret & Burden & Stability \\
\midrule
Stagewise & 0.6026 & 0.1645 & 16.63 & 0.000 \\
Plug-in & 0.6246 & 0.1424 & 16.98 & 0.000 \\
DR value only & 0.6558 & 0.1112 & 14.01 & 0.139 \\
DR-LCB 0.50 & 0.6562 & 0.1108 & 13.97 & 0.139 \\
CASP 0.05 & 0.6568 & 0.1103 & 13.77 & 0.135 \\
Ma-style OPL & 0.6532 & 0.1138 & 14.63 & 0.028 \\
Wang-style generator & 0.6278 & 0.1392 & 16.74 & 0.357 \\
Oracle & 0.7670 & 0.0000 & 15.08 & 0.000 \\
\bottomrule
\end{tabular}
\caption{Block 2 full comparator table. For Blocks 2--5 the entries are sweep averages; Block 1 is reported directly.}
\label{tab:app-block2-full}
\end{table}

\begin{table}[!htbp]
\centering
\small
\begin{tabular}{@{}lrrrr@{}}
\toprule
Method & Value & Regret & Burden & Stability \\
\midrule
Stagewise & 0.7271 & 0.1308 & 37.84 & 0.000 \\
Plug-in & 0.7717 & 0.0862 & 44.60 & 0.000 \\
DR value only & 0.7649 & 0.0929 & 30.18 & 0.170 \\
DR-LCB 0.50 & 0.7645 & 0.0934 & 29.03 & 0.153 \\
CASP 0.05 & 0.7572 & 0.1007 & 21.02 & 0.111 \\
Ma-style OPL & 0.7722 & 0.0856 & 38.44 & 0.052 \\
Wang-style generator & 0.7430 & 0.1148 & 44.26 & 0.347 \\
Oracle & 0.8579 & 0.0000 & 41.97 & 0.000 \\
\bottomrule
\end{tabular}
\caption{Block 3 full comparator table. For Blocks 2--5 the entries are sweep averages; Block 1 is reported directly.}
\label{tab:app-block3-full}
\end{table}

\begin{table}[!htbp]
\centering
\small
\begin{tabular}{@{}lrrrr@{}}
\toprule
Method & Value & Regret & Burden & Stability \\
\midrule
Stagewise & 0.7023 & 0.1658 & 76.42 & 0.000 \\
Plug-in & 0.7354 & 0.1328 & 86.25 & 0.000 \\
DR value only & 0.7321 & 0.1360 & 57.86 & 0.156 \\
DR-LCB 0.50 & 0.7335 & 0.1346 & 58.47 & 0.167 \\
CASP 0.05 & 0.7324 & 0.1358 & 51.64 & 0.104 \\
Ma-style OPL & 0.7383 & 0.1299 & 73.97 & 0.042 \\
Wang-style generator & 0.7123 & 0.1559 & 81.32 & 0.312 \\
Oracle & 0.8682 & 0.0000 & 72.80 & 0.000 \\
\bottomrule
\end{tabular}
\caption{Block 4 full comparator table. For Blocks 2--5 the entries are sweep averages; Block 1 is reported directly.}
\label{tab:app-block4-full}
\end{table}

\begin{table}[!htbp]
\centering
\small
\begin{tabular}{@{}lrrrr@{}}
\toprule
Method & Value & Regret & Burden & Stability \\
\midrule
Stagewise & 0.6763 & 0.1486 & 24.70 & 0.000 \\
Plug-in & 0.7170 & 0.1079 & 32.09 & 0.000 \\
DR value only & 0.7176 & 0.1074 & 23.02 & 0.181 \\
DR-LCB 0.50 & 0.7159 & 0.1091 & 22.58 & 0.188 \\
CASP 0.05 & 0.7125 & 0.1124 & 20.66 & 0.146 \\
Ma-style OPL & 0.7226 & 0.1024 & 27.87 & 0.062 \\
Wang-style generator & 0.7010 & 0.1240 & 31.67 & 0.368 \\
Oracle & 0.8250 & 0.0000 & 28.44 & 0.000 \\
\bottomrule
\end{tabular}
\caption{Block 5 full comparator table. For Blocks 2--5 the entries are sweep averages; Block 1 is reported directly.}
\label{tab:app-block5-full}
\end{table}

\subsection{CASP Ablation Diagnostics}

\begin{table*}[!t]
\centering
\small
\begin{tabular}{@{}llrrrr@{}}
\toprule
Block & Ablation & Avg.\ value & Avg.\ regret & Avg.\ burden & Avg.\ stability \\
\midrule
B2 & Normalized full & 0.6587 & 0.1084 & 13.32 & 0.151 \\
B2 & Normalized stage-1 only & 0.6560 & 0.1110 & 13.96 & 0.139 \\
B2 & Normalized stage-2 only & 0.6588 & 0.1082 & 13.38 & 0.151 \\
B2 & Raw full & 0.6586 & 0.1084 & 12.15 & 0.591 \\
B3 & Normalized full & 0.7383 & 0.1196 & 14.94 & 0.104 \\
B3 & Normalized stage-1 only & 0.7585 & 0.0993 & 18.71 & 0.135 \\
B3 & Normalized stage-2 only & 0.7397 & 0.1182 & 23.14 & 0.132 \\
B3 & Raw full & 0.6686 & 0.1893 & 9.02 & 0.667 \\
B4 & Normalized full & 0.7153 & 0.1528 & 34.20 & 0.125 \\
B4 & Normalized stage-1 only & 0.7355 & 0.1326 & 52.93 & 0.177 \\
B4 & Normalized stage-2 only & 0.7172 & 0.1510 & 38.33 & 0.125 \\
B4 & Raw full & 0.6759 & 0.1923 & 21.59 & 0.844 \\
B5 & Normalized full & 0.6947 & 0.1302 & 14.28 & 0.132 \\
B5 & Normalized stage-1 only & 0.7085 & 0.1164 & 17.20 & 0.132 \\
B5 & Normalized stage-2 only & 0.7013 & 0.1236 & 18.76 & 0.160 \\
B5 & Raw full & 0.6467 & 0.1782 & 8.42 & 0.757 \\
\bottomrule
\end{tabular}
\caption{CASP ablation summary across Blocks 2--5. The normalized full burden is the main ablation baseline; the stage-1-only and stage-2-only variants isolate which component of the support burden carries the empirical signal; the raw full burden documents the pre-calibration failure mode.}
\label{tab:app-ablation-summary}
\end{table*}

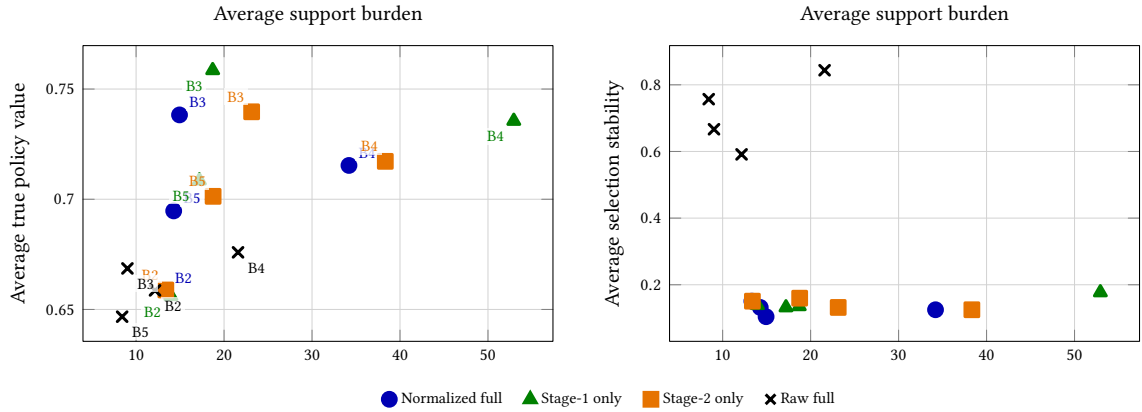
\begin{figure*}[t]
\centering
\begin{tikzpicture}
\begin{groupplot}[
    caspgroup,
    group style={group size=2 by 1, horizontal sep=1.55cm},
    width=0.44\textwidth,
    height=0.31\textwidth,
    xlabel={Average support burden},
    every axis plot/.append style={only marks, mark size=2.6pt},
]
\nextgroupplot[
    ylabel={Average true policy value},
    legend to name=ablationlegend,
    legend columns=4,
    legend style={font=\scriptsize, draw=none, fill=none, /tikz/every even column/.append style={column sep=0.7em}},
]
\addplot[blue!70!black, mark=*,
    nodes near coords={\pgfplotspointmeta},
    point meta=explicit symbolic,
    every node near coord/.append style={font=\scriptsize, anchor=south west, xshift=2pt, yshift=1pt, text=blue!70!black, fill=white, fill opacity=0.75, text opacity=1, inner sep=1.2pt}
] table[x=normalized_full_burden,y=normalized_full_value,meta=block_label,col sep=comma]{figures/phase3_external_baselines/ablation_frontier_key.csv};
\addlegendentry{Normalized full}
\addplot[green!50!black, mark=triangle*,
    nodes near coords={\pgfplotspointmeta},
    point meta=explicit symbolic,
    every node near coord/.append style={font=\scriptsize, anchor=north east, xshift=-2pt, yshift=-2pt, text=green!50!black, fill=white, fill opacity=0.75, text opacity=1, inner sep=1.2pt}
] table[x=stage1_only_burden,y=stage1_only_value,meta=block_label,col sep=comma]{figures/phase3_external_baselines/ablation_frontier_key.csv};
\addlegendentry{Stage-1 only}
\addplot[orange!90!black, mark=square*,
    nodes near coords={\pgfplotspointmeta},
    point meta=explicit symbolic,
    every node near coord/.append style={font=\scriptsize, anchor=south east, xshift=-1pt, yshift=2pt, text=orange!90!black, fill=white, fill opacity=0.75, text opacity=1, inner sep=1.2pt}
] table[x=stage2_only_burden,y=stage2_only_value,meta=block_label,col sep=comma]{figures/phase3_external_baselines/ablation_frontier_key.csv};
\addlegendentry{Stage-2 only}
\addplot[black, mark=x,
    nodes near coords={\pgfplotspointmeta},
    point meta=explicit symbolic,
    every node near coord/.append style={font=\scriptsize, anchor=north west, xshift=2pt, yshift=-2pt, text=black, fill=white, fill opacity=0.75, text opacity=1, inner sep=1.2pt}
] table[x=raw_full_burden,y=raw_full_value,meta=block_label,col sep=comma]{figures/phase3_external_baselines/ablation_frontier_key.csv};
\addlegendentry{Raw full}
\nextgroupplot[
    ylabel={Average selection stability},
]
\addplot[blue!70!black, mark=*] table[x=normalized_full_burden,y=normalized_full_stability,col sep=comma]{figures/phase3_external_baselines/ablation_frontier_key.csv};
\addplot[green!50!black, mark=triangle*] table[x=stage1_only_burden,y=stage1_only_stability,col sep=comma]{figures/phase3_external_baselines/ablation_frontier_key.csv};
\addplot[orange!90!black, mark=square*] table[x=stage2_only_burden,y=stage2_only_stability,col sep=comma]{figures/phase3_external_baselines/ablation_frontier_key.csv};
\addplot[black, mark=x] table[x=raw_full_burden,y=raw_full_stability,col sep=comma]{figures/phase3_external_baselines/ablation_frontier_key.csv};
\end{groupplot}
\path (group c1r1.south west) -- (group c2r1.south east)
node[midway,below=12pt] {\pgfplotslegendfromname{ablationlegend}};
\end{tikzpicture}
\caption{Cross-block CASP ablation frontier over Blocks 2--5. The normalized full burden is the main ablation baseline; the stage-1-only and stage-2-only variants show how much of the empirical signal survives when one component of the burden is removed; the raw full variant illustrates the pre-calibration failure mode.}
\label{fig:app-ablation-frontier}
\Description{Two-panel scatter plot for CASP ablations across Blocks 2 through 5. The left panel shows average value against average support burden for four ablation variants: normalized full, stage-1-only, stage-2-only, and raw full, with each point labeled by block. The raw full variant achieves the lowest burden but much weaker value in Blocks 3 through 5, while the normalized variants occupy higher-value regions. The right panel shows average selection stability against average burden for the same ablations; the raw full variant is highly mode-concentrated, indicating overly aggressive conservatism, while the normalized variants are less collapsed.}
\end{figure*}

The ablation results support two practical conclusions. First, the raw full burden is too aggressive in the harder blocks: it drives burden down, but it does so by collapsing onto highly repetitive conservative policies with substantially worse value. Second, once the penalty is normalized, much of the empirical benefit survives. In Blocks~3--5 the stage-1-only ablation is the strongest partial-burden variant, which suggests that a large share of the useful conservatism enters through the stage-1 support pathway. The full normalized burden remains the default because it is the most methodologically aligned version of \CASP{}, but the ablations help explain where its empirical signal comes from.

\subsection{Block 5 Precision Verification Addendum}

The main empirical narrative is intentionally not driven by post hoc reruns. Still, because the sample-size block sits close to the deployment-facing frontier that matters most for interpretation, we include a verification-focused follow-up restricted to the higher-sample-size segment $n\in\{600,1200,2400,4800\}$. This check doubles the replication count relative to the original Block~5 design. Its purpose is to tighten the sample-size evidence, not to overwrite the primary external-baseline comparison.

\begin{table}[!htbp]
\centering
\small
\begin{tabular}{@{}lrrrr@{}}
\toprule
Method & Avg.\ value & Avg.\ regret & Avg.\ burden & Avg.\ stability \\
\midrule
Stagewise & 0.6840 & 0.1432 & 26.57 & 0.000 \\
Plug-in & 0.7375 & 0.0896 & 32.83 & 0.000 \\
DR value only & 0.7405 & 0.0867 & 23.33 & 0.135 \\
DR-LCB 0.50 & 0.7381 & 0.0891 & 22.25 & 0.130 \\
CASP 0.05 & 0.7355 & 0.0916 & 19.80 & 0.104 \\
Ma-style OPL & 0.7456 & 0.0815 & 28.60 & 0.036 \\
Wang-style generator & 0.7201 & 0.1071 & 30.45 & 0.312 \\
Oracle & 0.8272 & 0.0000 & 27.43 & 0.000 \\
\bottomrule
\end{tabular}
\caption{Verification-focused Block 5 precision rerun over sample sizes $n\in\{600,1200,2400,4800\}$. The rerun preserves the same qualitative story as the locked external-baseline suite: Ma-style OPL remains strongest on raw value, while \CASP{} remains the lowest-burden learned selector.}
\label{tab:block5-precision-verification}
\end{table}

The addendum confirms the same qualitative reading as the main sample-size block. \citet{ma2020twostage} OPL remains the strongest raw-value learned comparator in this regime, while \CASP{} remains the lowest-burden learned selector by a substantial margin. In the precision rerun, \CASP{} averages burden $19.80$ versus $28.60$ for \citet{ma2020twostage} OPL, while \citet{ma2020twostage} retains the raw-value lead ($0.7456$ versus $0.7355$). The rerun reinforces the deployment-aware value--burden interpretation without changing the main claim.

\subsection{Supplementary Application Tables}

This subsection reports additional application tables that complement the main results.

\begin{table*}[!t]
\centering
\small
\begin{tabular}{lrrrrrr}
\toprule
Comparator & Oracle & DR & Burden & ESS & Max $w$ & Mode freq. \\
\midrule
Behavior & 0.474 & 0.775 & 1.0 & 5000.0 & 1.0 & -- \\
Random uniform & 0.494 & 0.720 & $5.39\times 10^6$ & 31475.0 & 0.9 & -- \\
Reconstructed oracle & 0.543 & 0.646 & $7.12\times 10^8$ & 226.4 & 75.4 & -- \\
Stagewise proxy & 0.465 & 0.845 & $7.64\times 10^8$ & 913.4 & 77.2 & -- \\
Plug-in reward & 0.442 & 0.888 & $6.58\times 10^8$ & 352.4 & 78.8 & -- \\
DR value only & 0.442 & 0.888 & $6.58\times 10^8$ & 352.4 & 78.8 & 1.00 \\
DR-LCB ($\beta=0.50$) & 0.442 & 0.888 & $6.58\times 10^8$ & 352.4 & 78.8 & 1.00 \\
\CASP{} ($\lambda=0.05$) & 0.453 & 0.831 & 39.6 & 212.9 & 78.3 & 0.55 \\
Ma-style OPL & 0.445 & 0.886 & $6.64\times 10^8$ & 405.0 & 78.8 & 0.10 \\
Wang-style generator & 0.446 & 0.898 & $6.41\times 10^8$ & 572.8 & 77.6 & 0.90 \\
\bottomrule
\end{tabular}
\caption{Full application comparator summary for the accepted reconstructed \texttt{MovieLens 1M} run, including effective-sample-size and tail-weight diagnostics.}
\label{tab:app-full-comparators}
\end{table*}

\begin{table*}[!t]
\centering
\small
\begin{tabular}{lrrrrrr}
\toprule
Ablation & Oracle & DR & Burden & ESS & Max $w$ & Unique policies \\
\midrule
Normalized full & 0.453 & 0.831 & 39.6 & 212.9 & 78.3 & 5 \\
Normalized stage-1 only & 0.453 & 0.831 & 39.6 & 212.9 & 78.3 & 4 \\
Normalized stage-2 only & 0.442 & 0.888 & $6.58\times 10^8$ & 352.4 & 78.8 & 1 \\
Raw full & 0.519 & 0.816 & 35.8 & 421.5 & 72.9 & 6 \\
\bottomrule
\end{tabular}
\caption{Application ablation summary at $\lambda=0.05$. The raw-full variant is retained as an appendix calibration sensitivity rather than a headline main-text result.}
\label{tab:app-ablation}
\end{table*}

\begin{table*}[!t]
\centering
\small
\begin{tabular}{rrrrrr}
\toprule
$\lambda$ & DR value & Burden & ESS & Max $w$ & Mode freq. \\
\midrule
0.000 & 0.888 & $6.58\times 10^8$ & 352.4 & 78.8 & 1.00 \\
0.010 & 0.831 & 39.6 & 212.9 & 78.3 & 0.70 \\
0.020 & 0.831 & 39.6 & 212.9 & 78.3 & 0.70 \\
0.050 & 0.831 & 39.6 & 212.9 & 78.3 & 0.55 \\
0.100 & 0.831 & 39.6 & 212.9 & 78.3 & 0.50 \\
0.200 & 0.830 & 39.4 & 229.6 & 77.2 & 0.45 \\
\bottomrule
\end{tabular}
\caption{\CASP{} $\lambda$-sensitivity on the reconstructed \texttt{MovieLens 1M} application. The grid is generated from the same completed application run as the main comparator table.}
\label{tab:app-lambda-sensitivity}
\end{table*}

\begin{table}[!htbp]
\centering
\small
\begin{tabular}{lrrr}
\toprule
Comparator & Burden & ESS & Max $w$ \\
\midrule
\CASP{} ($\lambda=0.05$) & 39.6 & 212.9 & 78.3 \\
DR value only & $6.58\times 10^8$ & 352.4 & 78.8 \\
DR-LCB ($\beta=0.50$) & $6.58\times 10^8$ & 352.4 & 78.8 \\
Ma-style OPL & $6.64\times 10^8$ & 405.0 & 78.8 \\
Wang-style generator & $6.41\times 10^8$ & 572.8 & 77.6 \\
\bottomrule
\end{tabular}
\caption{Application effective-sample-size and tail-weight diagnostics. These diagnostics separate the support-burden signal from conventional ESS and maximum-weight summaries, which are reported for calibration rather than used as the headline claim.}
\label{tab:app-weight-diagnostics}
\end{table}

\begin{table*}[!t]
\centering
\small
\begin{tabular}{lrr}
\toprule
Selector & Burden & Floor-implied off-support mass \\
\midrule
\CASP{} ($\lambda=0.05$) & 39.6 & $<0.001$ \\
DR value only & $6.58\times 10^8$ & 0.658 \\
DR-LCB ($\beta=0.50$) & $6.58\times 10^8$ & 0.658 \\
Ma-style OPL & $6.64\times 10^8$ & 0.664 \\
Wang-style generator & $6.41\times 10^8$ & 0.641 \\
Stagewise proxy & $7.64\times 10^8$ & 0.764 \\
Raw full & 35.8 & $<0.001$ \\
\bottomrule
\end{tabular}
\vspace{0.7em}
\begin{tabular}{lrrr}
\toprule
Generator & Zero stage-1 support share & DR-only selected share & \CASP{} selected share \\
\midrule
Popularity & 0.165 & 0.076 & 0.561 \\
Genre & 0.870 & 0.000 & 0.119 \\
Collaborative & 0.693 & 0.924 & 0.307 \\
LongTail & 0.861 & 0.000 & 0.013 \\
\bottomrule
\end{tabular}
\caption{Application support-violation diagnostics. The first panel reports the burden-scale implication of the $10^{-9}$ denominator floor: values near $0.65$ mean the selector frequently asks for action pairs outside reconstructed support. The second panel explains the main source of that behavior: DR-value-only concentrates on the collaborative generator, which has zero reconstructed stage-1 support on many evaluation contexts, while \CASP{} redistributes mass toward supported generators.}
\label{tab:app-support-violation}
\end{table*}

\begin{table*}[!t]
\centering
\small
\begin{tabularx}{\textwidth}{@{}Xrrrrr@{}}
\toprule
Variant & CASP DR & CASP burden & DR-only burden & Ma burden & CASP ESS \\
\midrule
Baseline cached L=30, rating >=4 & 0.831 & 39.6 & $6.58\times 10^8$ & $6.64\times 10^8$ & 212.9 \\
Derived L=20 candidate truncation with support refiltering & 0.842 & 24.9 & $7.09\times 10^8$ & $7.44\times 10^8$ & 253.3 \\
Reward relabeling rating >=5 on accepted support & 0.495 & 40.0 & $6.94\times 10^8$ & $6.94\times 10^8$ & 232.0 \\
Stage-1 logging smoothing +0.20 on accepted support & 0.824 & 41.1 & $6.68\times 10^8$ & $6.70\times 10^8$ & 158.9 \\
Stage-2 logging smoothing +0.10 within feasible sets & 0.831 & 39.8 & $6.58\times 10^8$ & $6.63\times 10^8$ & 213.9 \\
\bottomrule
\end{tabularx}
\caption{Full cached MovieLens robustness suite. Each row reruns the full comparator layer on a cached or derived local sensitivity variant of the accepted reconstructed pool.}
\label{tab:app-robustness-full}
\end{table*}

\begin{table*}[!t]
\centering
\small
\begin{tabular}{lc}
\toprule
Diagnostic & Value \\
\midrule
Final paper-facing contexts & 25,000 \\
Strict $\ge 2$-support share & 0.36388 \\
Dominant stage-1 share & 0.55368 \\
Minimum generator share & 0.10932 \\
Mean support-generator count & 1.411 \\
Share with exactly two support generators & 0.399 \\
Mean stage-1 entropy on eval slice & 0.2828 \\
Maximum support-generator count & 3 \\
Fallback singleton pool mode & blended\_strict\_plus\_capped\_diverse\_singleton\_fallback \\
\bottomrule
\end{tabular}
\caption{Application support and logging diagnostics for the accepted reconstructed \texttt{MovieLens 1M} pool.}
\label{tab:app-support}
\end{table*}

\end{document}